\documentclass{siamonline220329}

\usepackage{amsfonts,amssymb}
\usepackage{enumerate}
\usepackage{nameref}
\usepackage[X2,T2A]{fontenc}

\usepackage{comment}
\usepackage[frozencache]{minted}
\usepackage{xpatch}
\usepackage{letltxmacro}

\usepackage{enumitem}
\setlist[enumerate]{leftmargin=.5in} % Prevent itemized lists from running into the left margin inside theorems and= proofs
\setlist[itemize]{leftmargin=.5in}

\usepackage{bm}
\usepackage{mathdots}
\usepackage{mathtools}
\usepackage{algorithmic}
\Crefname{ALC@unique}{Line}{Lines} % To \cref lines in the algorithm
\usepackage{soul} % for \ul command
\usepackage{appendix}

\newsiamthm{example}{Example}
\newsiamremark{remark}{Remark}
\newsiamremark{assumption}{Assumption}

\newcommand{\bit}{\begin{itemize}}
\newcommand{\eit}{\end{itemize}}
\newcommand{\ben}{\begin{enumerate}}
\newcommand{\een}{\end{enumerate}}

\newcommand {\real} {\mathbb{R}}

\newcommand{\bA}{\ensuremath{\mathbf{A}}}
\newcommand{\bB}{\ensuremath{\mathbf{B}}}

\newcommand{\bD}{\ensuremath{\mathbf{D}}}

\newcommand{\bH}{\ensuremath{\mathbf{H}}}

\newcommand{\bN}{\ensuremath{\mathbf{N}}}

\newcommand{\bP}{\ensuremath{\mathbf{P}}}

\newcommand{\bS}{\ensuremath{\mathbf{S}}}

\newcommand{\bW}{\ensuremath{\mathbf{W}}}

\newcommand{\bb}{\ensuremath{\mathbf{b}}}
\newcommand{\bc}{\ensuremath{\mathbf{c}}}

\renewcommand{\bf}{\ensuremath{\mathbf{f}}}

\newcommand{\bp}{\ensuremath{\mathbf{p}}}
\newcommand{\bq}{\ensuremath{\mathbf{q}}}

\newcommand{\bu}{\ensuremath{\mathbf{u}}}
\newcommand{\bv}{\ensuremath{\mathbf{v}}}
\newcommand{\bw}{\ensuremath{\mathbf{w}}}
\newcommand{\bx}{\ensuremath{\mathbf{x}}}

\newcommand {\btheta} {\mbox{\boldmath $\theta$}}

\newcommand{\cF}{\ensuremath{\mathcal{F}}}

\newcommand{\cM}{\ensuremath{\mathcal{M}}}

\newcommand{\Qbee}{\textsf{QBee} }

\usepackage[textsize=tiny]{todonotes}

\newcommand{\nx}{{n}} % Spatial discretization dimension
\newcommand{\nd}{{n_d}} % Number of dependent variables
\newcommand{\nn}{{N}}   % Total system dimension

\title{Exact and optimal quadratization of nonlinear finite-dimensional \\ non-autonomous dynamical systems
\thanks{
\funding{B.K. is supported in part by NSF-CMMMI award 2144023. B.K. would like to thank the Isaac Newton Institute for Mathematical Sciences, Cambridge, for support and hospitality during the programme ``The mathematical and statistical foundation of future data-driven engineering" where part of the work on this paper was undertaken (supported by EPSRC grant EP/R014604/1).
Part of work has been done during the visit of B.K. to \'Ecole Polytechnique within the PANTOMIME project funded by AAP INS2I CNRS.
G.P. was partially supported by the PANTOMIME project, the Paris Ile-de-France region, and NSF grants DMS-1760448 and DMS-1853650.}}}
\author{
Andrey Bychkov\thanks{Thales Research \& Technology , France (\email{abychkov\_edu@proton.me})}
\and 
Opal Issan\thanks{University of Calfornia San Diego, La Jolla, CA, (\email{oissan@ucsd.edu})}
\and Gleb Pogudin\thanks{LIX, CNRS, \'{E}cole Polytechnique, Institute Polytechnique de Paris, France (\email{gleb.pogudin@polytechnique.edu})}
\and
Boris Kramer\thanks{University of Calfornia San Diego, La Jolla, CA, (\email{bmkramer@ucsd.edu})}
}

\headers{Exact and optimal quadratization of nonlinear non-autonomous ODEs}{A. Bychkov, O. Issan, G. Pogudin and B. Kramer}

\begin{document}
\maketitle
%%%%%%%%%%%

\begin{abstract}
Quadratization of polynomial and nonpolynomial systems of ordinary differential equations is advantageous in a variety of disciplines, such as systems theory, fluid mechanics, chemical reaction modeling and mathematical analysis.  A quadratization reveals new variables and  structures of a model, which may be easier to analyze, simulate, control, and provides a convenient parametrization for learning. 
This paper presents novel theory, algorithms and software capabilities for quadratization of non-autonomous ODEs. We provide existence results, depending on the regularity of the input function, for cases when a quadratic-bilinear system can be obtained through quadratization. 
We further develop existence results and an algorithm that generalizes the process of quadratization for systems with arbitrary dimension that retain the nonlinear structure when the dimension grows. For such systems, we provide dimension-agnostic quadratization.  An example is semi-discretized PDEs, where the nonlinear terms remain symbolically identical when the discretization size increases. 
As an important aspect for practical adoption of this research, we extended the capabilities of the \Qbee software towards both non-autonomous systems of ODEs and ODEs with arbitrary dimension. 
We present several examples of ODEs that were previously reported in the literature, and where our new algorithms find quadratized ODE systems with lower dimension than the previously reported lifting transformations. 
We further highlight an important area of quadratization: reduced-order model learning. This area can benefit significantly from working in the optimal lifting variables, where quadratic models provide a direct parametrization of the model that also avoids additional hyperreduction for the nonlinear terms. A solar wind example highlights these advantages. 
\end{abstract}



%%%%%%%%%%%%%%%%%%%%%%%%%%%%%%%%%%%%%%%%%%%%%%%%%%%%%%%%%
\section{Introduction}
%%%%%%%%%%%%%%%%%%%%%%%%%%%%%%%%%%%%%%%%%%%%%%%%%%%%%%%%
Evolutionary processes in engineering and science are often modeled with nonlinear non-autonomous ordinary differential equations that describe the time evolution of the states of the system, i.e., the physically necessary and relevant variables. However, these models are not unique: the same evolutionary process can be modeled with different variables, which can have a tremendous impact on computational modeling and analysis.
This idea of variable transformation (referred to as \textit{lifting} when extra variables are added) to promote model structure is found across different communities, with literature spanning half a century. We first review general variable transformations and then focus on the appealing quadratic form. 

In fluid dynamics, variable transformations have long been recognized as providing useful alternative representations. While the Euler and Navier-Stokes equations are most commonly derived in conservative variables, where each state is a conserved quantity (mass, momentum, energy), symmetric variables have been exploited to guarantee stable models~\cite{hughes1986new}. Additionally, recent work by~\cite{halpern_2021_antisymmetric} shows that the fluid and plasma equations can be reformulated via an anti-symmetric variable transformation to retain robust conservation properties in numerical simulations.  Stability-preserving inner products for projection-based reduced-order models (ROMs) hint at variable transformations as a proper choice to increase stability for nonlinear fluid models, see \cite{kalashnikova2011stable,rezaian2020impact}. As another classic example, the well-known Cole-Hopf transformation turns the nonlinear Burgers' partial differential equation (PDE) into a linear PDE~\cite{cole1951quasi,hopf1950partial}. 
In the dynamical systems field, the Koopman operator is a linear infinite-dimensional operator that describes the dynamics of observables of nonlinear systems. Dynamic mode decomposition~\cite{rowley2009spectral,schmid2010dynamic} approximates the spectrum of the Koopman operator, and the choice of observables has a significant impact on the quality of the learned model, which led to the extended DMD algorithm~\cite{williams2015data,netto2021analytical}.
For control design, the idea of transforming a general nonlinear system into a system with more structure is also common practice: the concept of feedback linearization transforms a general nonlinear system into a structured linear model \cite{jakubczyk1980linearization,Khalil_NonlinearSystems}. This is done via a nonlinear state transformation, where the transformed state might be augmented, i.e., have increased dimension relative to the original state. The lifting transformations known in feedback linearization are specific to the desired target model form.
A change of variables can also ensure that physical constraints are more easily met in a simulation, see \cite{nam2011space,hassler2020finite}. 
Bringing nonlinear systems into canonical and abstract forms can then improve analysis, as seen in~\cite{liu2015abstraction,brenig2018reducing}. 
The authors in~\cite{savageau1987recasting} showed that all ODE systems with (nested) elementary functions can be recast in a special polynomial system form, which is then faster to solve numerically.

Quadratic model structure, as a special case, has seen the broad interest due to the ease of working with quadratic models. 
In the context of optimization, McCormick~\cite{mccormick1976computability} is credited with first introducing variable substitutions to achieve quadratic structure so that non-convex optimization problems can be recast as convex problems in the new variables.
In fluid mechanics, the quadratic model structure of the specific volume variable representation has been exploited in~\cite{balajewicz2016minimal} to allow for model stabilization.
To analyze equilibrium branches of geometrically nonlinear finite element models, the authors in \cite{guillot2019generic} recast the model into a quadratic form, for which the Jacobian and a specific Taylor series can be easily obtained. They then use the Asymptotic Numerical Method to find the equilibrium branches.
In the area of analog computing with chemical reaction networks, quadratic forms correspond to the notion of elementary chemical reactions. Transforming an arbitrary polynomial ODE system into a quadratic one can be used to establish the Turing completeness of elementary chemical reactions~\cite{bournez2007polynomial,fages2017strong,lifeware1}.
Lifted models in quadratic and quadratic-bilinear (QB) form have seen great interest in the systems and control community in the past decade, as the QB structure is appealing for intrusive model reduction~\cite{gu2011qlmor,bennerBreiten2015twoSided,bennergoyal2016QBIRKA,KW18nonlinearMORliftingPOD,KW2019_balanced_truncation_lifted_QB}. Those methods project the quadratic operators of a high-dimensional lifted model onto a reduced space to obtain a quadratic reduced model. The quadratic model structure eliminates the need for additional hyperreduction/interpolation to reduce the expense of evaluating the nonlinear term~\cite{deim2010,astrid2008missing,barrault2004empirical,carlberg2013gnat,nguyen2008best}, where in some cases the number of interpolation points required for accuracy eliminates computational gains~\cite{bergmann2009enablers,huangAIAA18RomRocketCombustion}. Moreover, the quadratic model structure has fewer degrees of freedom in the ROM than a cubic or high-order polynomial model, making quadratic models the most desired lifted form. 
The \textit{Lift \& Learn} method \cite{QKPW2020_lift_and_learn,QKMW2019_transform_and_learn} and related work~
\cite{SKHW2020_learning_ROMs_combustor,gosea2018data,JQK2021_performanceCompCombustion,McQuarrie_regularizedOPINF}, leverage lifting transformations to learn low-order polynomial ROMs of complex nonlinear systems, such as combustion dynamics, from lifted data. For quadratic and cubic model structures, one can equip these learned ROMs with stability guarantees, see~\cite{K2020_stability_domains_QBROMs,SKP_PIregulartizationOPINF}.

With such great appeal and interest in quadratic and polynomial dynamical system structure, it is natural to ask: Which systems can be brought into polynomial (via \emph{polynomialization}) and further into quadratic (via \emph{quadratization}) form? Which algorithms and methods exist to perform such transformations?
The theoretical results that any system written using (nested) elementary functions can be polynomialized and every polynomial system can be quadratized have been discovered and established in different communities and contexts, see~\cite{Appelroth1902,Lagutinskii,kerner1981universal,CPSW05,gu2011qlmor,carravetta2015global,carravetta2020solution}.
The proofs of these theorems are constructive and can be turned into algorithms. However, a straightforward approach would require introducing an excessively large number of variables and thus making the resulting dynamical system hard to use and/or analyze.
The resulting quadratic systems can be pruned using constrained programming techniques, and designed algorithms and methods are presented for polynomialization in~\cite{lifeware2} and quadratization in~\cite{lifeware1}. Both algorithms are implemented in the \textsf{Biocham} software~\cite{Biocham}.
We take a different approach in~\cite{Bychkov2021} where the quadratization problem for a polynomial system is framed as a tree exploration.  The quadratization is obtained by building up the tree structure rather than pruning it as in~\cite{lifeware1}. 
Attractive features of the resulting algorithm \Qbee \cite{QBee} include optimality guarantees for the dimension of the lifted system and good performance in practice~\cite[Table~3]{Bychkov2021}. 
Although the previous version of \Qbee performed well on a selection of benchmark models from synthetic biology and some academic examples, see~\cite{Bychkov2021}, we found its applicability to engineering problems quite limited.  
In particular, dynamical systems models of engineering systems share common features that the previous version of \Qbee cannot handle:
\ben
  \item Engineering models are often driven by time-dependent inputs or controls, i.e., they are non-autonomous. Currently, there is no suitable theory for quadratization of such systems, and the previous version of \Qbee could not handle this.
  \item Semi-discretization of PDEs produces a system of ODEs with $\nx$ unknowns in the discretization. For such models, it would be computationally advantageous to find uniform quadratization scheme valid for any $\nx$.
  \item Many models are not polynomial (e.g., involve fractions or exponential functions), so a polynomialization procedure as in \textsf{Biocham}~\cite{lifeware2} is missing in \Qbee.
\een

There are four main contributions of this paper that address these challenges.
First, we provide new theorems (with constructive proofs) of existence of quadratizations for different classes of models. These theorems lead us to develop practical algorithms. 
Second, we extend the algorithms and implementation in the previous version of \Qbee with extra functionality to  (i) optimally quadratize polynomial systems with time-dependent inputs,  (ii) to find dimension-agnostic quadratizations for systems with variable dimension (e.g., semi-discretized PDEs), and (iii) to polynomialize and then quadratize non-polynomial terms. 
This functionality is available in \Qbee version 0.8 \cite{QBee}.
Third, we demonstrate how this new functionality addresses the three challenges above on models from chemical engineering, rocket combustion and space weather. These novel quadratizations from \Qbee outperform reported results in the literature, as we show in the corresponding sections.
Fourth, we present a numerical study that shows how these lifting transformations can create better coordinate systems to learn ROMs from data, a key use case of quadratization.

% Structure of the paper
This paper is organized as follows. \Cref{sec:poly} provides background on quadratization of autonomous polynomial ODEs \textit{of fixed dimension}. \Cref{sec:theory_inputs} presents new existence results and constructive proofs for quadratization of \textit{non-autonomous} polynomial models. \Cref{sec:agnostic} then presents new results for uniform quadratizations of families of ODEs arising from semi-discretization of PDEs. 
\Cref{sec:Qbee} outlines the \Qbee algorithm and its new capabilities. In \cref{sec:applications}, we showcase the results of our algorithm on models from chemical engineering and rocket combustion.
\Cref{sec:heliospheric} demonstrates the advantages of quadratization for learning reduced-order models on a problem from solar wind velocity prediction. 
Finally, \cref{sec:conclusions} offers conclusions and an outlook towards future work.

%%%%%%%%%%%%%%%%%%%%%%%%%%%%%%%%%%%%%%%%%%%%%%%%%%%%%%%%%%
\section{Notation and background}  \label{sec:poly}
%%%%%%%%%%%%%%%%%%%%%%%%%%%%%%%%%%%%%%%%%%%%%%%%%%%%%%%%%%
We start in \cref{ss:notation} with defining notation and then review known results about the existence of quadratizations in \cref{ss:QuadPolyODEs}. 
These results also provide an upper bound for the order of quadratization, i.e., the  maximum required number of variables needed for achieving a lifted quadratic model. 

%%%%%%%%%%%%%%%%%%%%%%%%%%%%%
\subsection{Notation and definitions} \label{ss:notation}
We denote by $\bx= [x_1, x_2, \ldots, x_\nn]^\top$ an $\nn$-dimensional column vector (in either $\mathbb{R}^\nn$ or $\mathbb{C}^\nn$), which generically represents the state of a dynamical system (we use different notation for the application problems where states have physical meaning). We denote with $\dot{\bx}$ its derivative (usually with respect to time, $t$, where $t>0$). Moreover, we denote by $\bu=\bu(t) \in \mathbb{R}^r$ a generic input vector. We often omit the explicit dependence of $t$ for ease of notation. 
The symbol $\odot$ denotes the component-wise product of vectors (the Hadamard product), and $\bx^\ell$ denotes the $\ell$th power of $\bx$, also understood component-wise. Moreover, $\otimes$ denotes the Kronecker product of vectors, e.g., $[x_1 \ x_2]^\top \otimes [y_1 \ y_2]^\top = [x_1 y_1 \  x_1 y_2 \ x_2 y_1 \  x_2 y_2]^\top$, 
and $\otimes^{\prime}$ is the compact Kronecker product, which removes redundant terms (only one term here, $x_1x_2 = x_2x_1$) in the standard Kronecker product, e.g., $[x_1 \ x_2]^\top \otimes^{\prime} [y_1 \ y_2]^\top = [x_1 y_1 \  x_1 y_2 \  x_2 y_2]^\top$.
For a matrix $\bA \in \mathbb{R}^{\nn \times \nn}$, its $(i,j)$th entry is $A_{ij}$.
The set of nonnegative integers is denoted as $\mathbb{Z}_{\geqslant 0}$.

A product of positive-integer powers of variables is referred to as \textit{monomial} (e.g., $x^5y$) and the total degree of a monomial is the sum of the powers of the variables appearing in it. A \textit{polynomial} is a sum of monomials, e.g., $x + y^2$.
By $\mathbb{C}[\bx]$ and $\mathbb{C}[\bx, \bw]$ we denote the sets of all polynomials with (possibly complex) coefficients in $\bx$ and $\bx, \bw$, respectively. The sets $\mathbb{C}[\bx,\bw,\bu]$ and $\mathbb{C}[\bx,\bw,\bu, \dot{\bu}]$ are sets of polynomials defined similarly.
Let $p(\bx)$ be a polynomial, then $\deg_{x_i} p(\bx)$ denotes the degree of $p$ with respect to $x_i$, that is, the maximal power of $x_i$ appearing in $p(\bx)$. Similarly, $\deg p(\bx)$  denotes the total degree of $p$, that is, the maximum of the total degrees of the monomials appearing in $p(\bx)$.
For example $\deg_x (x^5y) = 5$ and $\deg (x^5 y) = 6$.
The degree of a vector or a matrix is defined as the maximum of the degrees of its entries.

%%%%%%%%%%%%%%%%%%%%%%%%%%%%%%%%%%%%%%%%%%%%%%%%%%%
\subsection{Quadratization of autonomous polynomial ODEs} \label{ss:QuadPolyODEs}
We begin with a definition. 
\begin{definition}[Quadratization]\label{def:quadr}
  Consider a polynomial system of ODEs
  \begin{equation}\label{eq:sys_main}
  \dot{\bx} = \bp(\bx),
  \end{equation}
  where $\bx = \bx(t) = [x_1(t), \ldots, x_\nn(t)]^\top$ and 
  $\bp(\bx) = [p_1(\bx), \ldots, p_\nn(\bx)]^\top$ with $p_1, \ldots, p_\nn \in \mathbb{C}[\bx]$. Then an $\ell$-dimensional vector of new variables
  \begin{equation}\label{eq:quadr}
      \bw = \bw(\bx) \quad \in \ \mathbb{C}[\bx]^\ell
  \end{equation}
  is said to be a \emph{quadratization} of~\eqref{eq:sys_main} if there exist vectors $\bq_1(\bx, \bw)$ and $\bq_2(\bx, \bw)$ of dimensions $\nn$ and $\ell$, respectively, with the entries being polynomials of total degree at most two such that 
  \begin{equation}
  \dot{\bx} = \bq_1(\bx, \bw) \quad \text{ and }\quad \dot{\bw} = \bq_2(\bx, \bw)
  \end{equation}
  for every $\bx$ solving~\eqref{eq:sys_main}.
  The dimension $\ell$ of the vector $\bw$ is called the \emph{order of quadratization}.
  A quadratization of the smallest possible order is called an \emph{optimal quadratization}.
\end{definition}

\begin{example}\label{ex:quadr}
    Consider a two-dimensional system
    \begin{equation}\label{eq:quadr_ex}
    \dot{x}_1 = (x_1 + 1)^3 + x_2, \qquad \dot{x}_2 = x_1 + x_2
    \end{equation}
    We claim that a new variable $w(\bx) = (x_1 + 1)^2$ is a (polynomial) quadratization.
    Indeed, for the original states $x_1$ and $x_2$ we have
    \[
    \dot{x}_1 = (x_1 + 1)w + x_2, \qquad \dot{x}_2 = x_1 + x_2,
    \]
    so we take $\bq_1(\bx, w) = [(x_1 + 1)w + x_2, \; x_1 + x_2]^\top$.
    Furthermore,
    \[
    \dot{w} = 2(x_1 + 1)\dot{x}_1 = 2(x_1 + 1)^4 + 2(x_1 + 1)x_2 = 2w^2 + 2(x_1 + 1)x_2,
    \]
    so $q_2(\bx, w) = 2w^2 + 2(x_1 + 1)x_2$.
    Since this quadratization is of order one, it is optimal.
\end{example}

%%%%%%%%%%%%%%%%%%%%%%%%%%%%%%%%%%%%%%%%%%%%%%
\subsubsection{Monomial quadratization} 
This section considers the case where the new variables are restricted to be of monomial form. Those are useful in applications in synthetic biology~\cite{lifeware1} and are comparatively easier to find than polynomial terms because the search space is discrete and the problem of finding an optimal quadratization can be phrased as a combinatorial optimization problem.

\begin{definition}[Monomial quadratization]\label{def:monomial_quadr}
  If all the polynomials $w_1(\bx), \ldots, w_\ell(\bx)$ are monomials, the quadratization is called \emph{a monomial quadratization}.
  If a monomial quadratization of a system has the smallest possible order among all the monomial quadratizations of the system, it is called \emph{an optimal monomial quadratization}.
\end{definition}

\begin{example}[Continuation of~\Cref{ex:quadr}]
    The quadratization $w(\bx) = (x_1 + 1)^2$ introduced in~\Cref{ex:quadr} is not monomial.
    We next show that $w(\bx) = x_1^2$ is a monomial quadratization for the same system~\eqref{eq:quadr_ex}.
    For $x_1$ and $x_2$, we have
    \[
      \dot{x}_1 = x_1^3 + 3x_1^2 + 3x_1 + 1 + x_2 = x_1w + 3 w + 3x_1 + 1 + x_2,\quad \dot{x}_2 = x_1 + x_2,
    \]
    so $\bq_1(\bx, w) = [x_1 w + 3 w + 3x_1 + 1 + x_2,\; x_1 + x_2]^\top$.
    Furthermore,
    \[
    \dot{w} = 2(x_1 + 1)\dot{x_1} = 2x_1^4 + 8x_1^3 + 12 x_1^2 + 8 x_1 + 2 + (x_1 + 1)x_2 = 2w^2 + 8 x_1w + 12 w + 8 x_1 + 2 + (x_1 + 1)x_2, 
    \]
    so $q_2(\bx, w) = 2w^2 + 8 x_1w + 12 w + 8 x_1 + 2 + (x_1 + 1)x_2$. Therefore, $w(\bx) = x_1^2$ is an optimal monomial quadratization of~\eqref{eq:quadr_ex}.
    Here, both the monomial quadratization and the polynomial quadratization (cf.~\Cref{ex:arb_poly,ex:arb_poly2}) are optimal, yet the dynamics of the auxiliary state $w$ are a bit more involved in the monomial case. 
\end{example}

The next theorem assures that systems of polynomial ODEs are always guaranteed to have an exact quadratic representation in a different, possibly augmented, set of variables. 

\begin{theorem}[{\cite[Theorem~1]{CPSW05}}]\label{thm:existence_simple}
  For every ODE system of the form~\eqref{eq:sys_main}, there exists a monomial quadratization.
  Furthermore, if $d_i := \deg_{x_i}\bp(\bx)$, then the order of an optimal monomial quadratization does not exceed $\Pi_{i=1}^\nn (d_i + 1)$.
\end{theorem}

Monomial quadratizations have been actively studied in the literature. We summarize a few known results about the optimality, e.g., minimal order, of such quadratizations:
\begin{enumerate}
    \item In~\cite[Theorem~2]{lifeware1} it is shown that finding an optimal monomial quadratization is an NP-hard problem even if the polynomials are represented using dense encoding (that is, by all the coefficients up to certain degree).
    \item One natural approach (used, e.g., in~\cite{lifeware1}) to finding an optimal monomial quadratization is to take an explicit but large monomial quadratization constructed in the standard proof of \cref{thm:existence_simple} and search for the smallest subset of it which is a quadratization itself. 
    However, this approach can lead to a non-optimal quadratization as illustrated in \cite[Example 3]{Bychkov2021}.
    \item If one allows the new variables to be Laurent monomials (that is, have negative integer degrees), one can obtain a better upper bound for the number of new variables: it is always sufficient to take as many new variables as there are monomials in the right-hand side of the system~\cite[Proposition~1]{Bychkov2021}.
    However, we are not aware of any algorithm for finding an optimal Laurent monomial quadratization.
\end{enumerate}

%%%%%%%%%%%%%%%%%%%%%%%%%%%%%%%%%%%%%%%%%%%%%%
\subsubsection{Polynomial quadratization} 
Despite the appeal of monomial quadratizations in certain communities, it is natural to ask whether one may be able to obtain a substantially more concise quadratization if \emph{arbitrary polynomials} are allowed for the new variables (as in \cref{def:quadr}). 
Although, in the examples above, allowing arbitrary polynomials did not result in a lower-order quadratization, the following example from~\cite[Section 4]{Alauddin} shows that this is not always the case, and the difference between optimal monomial and arbitrary quadratizations may be significant. 

\begin{example}[Lower-order quadratization using arbitrary polynomials]\label{ex:arb_poly}
  Consider the scalar ODE
  \[
    \dot{x} = (x + 1)^k,
  \]
  where $k$ is a positive integer.
  A simple combinatorial argument~\cite[Section 4]{Alauddin} shows that at least $\frac{\sqrt{8k + 9} - 5}{2}$ monomial (that is, powers of $x$) variables are required for a monomial quadratization. Since $k$ can be arbitrarily large, this is concerning.
  In contrast, it is always possible to quadratize the ODE by adding one new variable $w(x) := (x + 1)^{k - 1}$ so that
  \begin{align*}
      \dot{x} & = (x + 1)w,\\
      \dot{w} & = \dot{x} (k - 1) (x + 1)^{k - 2} = (k - 1) (x + 1)^{2k - 2} = (k - 1)w^2.
  \end{align*}
\end{example}
Although the new variable $w$ above is per definition not a monomial, it can be thought of as a ``shifted monomial''. In \cite[Theorem~3.1]{Alauddin}, the authors state that if it is possible to quadratize a scalar ODE by a single new variable, then the new variable can always be taken to be such a ``shifted monomial''. However, if more new variables are added, the set of potentially useful new variables becomes richer as the next example~\cite[Theorem 3.2]{Alauddin} shows.

\begin{example}[Lower-order quadratization using arbitrary polynomials, continued]\label{ex:arb_poly2}
  Consider the scalar ODE of degree six:
  \begin{equation}\label{eq:deg6}
    \dot{x} = x^6 + x^4 + x^3.
  \end{equation}
  One can show using~\cite[Theorem~3.1]{Alauddin} that~\eqref{eq:deg6} cannot be quadratized using a single new variable. 
  In fact, if monomial new variables are used, at least three of them are needed~\cite[Lemma~5.7]{Alauddin}.
  However, one can quadratize~\eqref{eq:deg6} using just two new polynomial variables $w_1(x) := x^3$ and $w_2(x) := x^5 + \frac{5}{8}x^2$ as follows:
  \begin{align*}
      \dot{x} &= w_1^2 + xw_1 + w_1,\\
      \dot{w}_1 &= 3\left( w_1w_2 + x w_2 + \frac{3}{8}w_2 - \frac{5}{8}w_1 - \frac{15}{64}x^2\right),\\
      \dot{w}_2 &= 5\left( w_2^2 + w_1w_2 - \frac{9}{64}xw_1 - \frac{3}{8}w_2 + \frac{15}{64}x^2 \right).
  \end{align*}
\end{example}
  The previous examples show that there is merit in searching for polynomial quadratizations to limit the growth of additional variables needed to achieve quadratic form. 

%%%%%%%%%%%%%%%%%%%%%%%%%%%%%%%%%%%%%%%%%%%
\section{Quadratization of polynomial ODEs with inputs} \label{sec:theory_inputs}
%%%%%%%%%%%%%%%%%%%%%%%%%%%%%%%%%%%%%%%%%%%%%%
Many engineering systems are forced with external inputs or disturbances. This section presents new theoretical results for quadratization of polynomial ODEs with external forcing, which then lay the foundation for efficient algorithms.  
We are interested in bringing polynomial ODEs into quadratic-bilinear\footnote{In the systems and control literature, the terms $\bN_i\bx u_i$ above are called \emph{bilinear}---a subclass of \emph{control-affine systems}---as they are linear in the state $\bx$ and linear in the input/control $\bu$.} (QB) form, where
\begin{equation}\label{eq:qb_general}
 \dot{\bx} = \bA \bx + \bH (\bx \otimes \bx) + \sum_{i=1}^r \bN_i\bx u_i + \bB \bu,
\end{equation}
and $\bA\in \mathbb{R}^{\nn\times \nn}, \bH \in \mathbb{R}^{\nn\times \nn^2}$, $\bN_i \in \mathbb{R}^{\nn\times \nn}$ for $1 \leqslant i \leqslant r$, and $\bB \in \mathbb{R}^{\nn\times r}$. 
Quadratic-bilinear systems have seen great interest in the systems and control community, and are therefore an appealing target structure for the proposed quadratization methods, see the introduction for further discussion. 
We separate two cases for quadratization: \cref{sec:diffInputs} covers polynomial ODEs with differentiable inputs and \cref{sec:input-free} presents quadratization results for systems where the inputs $\bu(t)$ are not differentiable.

%

%%%%%%%%%%%%%%%%%%%%%%%%%%%%%%%%%%%%%%%%%%%%%%%%
\subsection{Systems with differentiable inputs} \label{sec:diffInputs}
We begin this section with a definition that extends \cref{def:quadr} to the non-autonomous polynomial case where the inputs are differentiable.
Since we aim at the QB form~\eqref{eq:qb_general}, our definitions are different from the ones in~\cite{carravetta2020solution} where the input variables are ignored when counting the degree.

\begin{definition}[Quadratization of non-autonomous polynomial ODEs]\label{def:quadr_input}
  Consider the system
  \begin{equation}\label{eq:sys_main_input}
  \dot{\bx} = \bp(\bx, \bu),
  \end{equation}
  where $\bx = \bx(t) = [x_1(t), \ldots, x_\nn(t)]^\top$ are the states, $\bu = \bu(t) = [u_1(t), \ldots, u_r(t)]^\top$ denote the inputs and 
  $p_1, \ldots, p_\nn \in \mathbb{C}[\bx, \bu]$.
  Then an $\ell$-dimensional vector of new variables
  \begin{equation}\label{eq:quadr_input}
      \bw = \bw(\bx, \bu) \quad \in \ \mathbb{C}[\bx, \bu]^\ell
  \end{equation}
 is said to be a \emph{quadratization} of~\eqref{eq:sys_main_input} if there exist vectors $\bq_1(\bx, \bw, \bu, \dot{\bu})$ and $\bq_2(\bx, \bw, \bu, \dot{\bu})$ of dimensions $\nn$ and $\ell$, respectively, with the entries being polynomials of total degree at most two such that 
  \begin{equation}
  \dot{\bx} = \bq_1(\bx, \bw, \bu, \dot{\bu}) \quad \text{ and }\quad \dot{\bw} = \bq_2(\bx, \bw, \bu, \dot{\bu})
  \end{equation}
  for every $\bx$ solving~\eqref{eq:sys_main_input}.
  The number $\ell$ is called the \emph{order of quadratization}.
  A quadratization of the smallest possible order is called an \emph{optimal quadratization}. 
  A \emph{monomial quadratization} of non-autonomous systems is defined similarly as in \cref{def:monomial_quadr}.
\end{definition}

\begin{example}\label{ex:input}
    Consider a non-autonomous scalar ODE $\dot{x} = x + x^2u$.
    We claim that $w(x, u) = xu$ is a quadratization.
    Indeed, the corresponding quadratic system is
    \[
      \dot{x} = x + xw, \quad \dot{w} = xu' + (x + x^2u)u = xu' + xu + w^2,
    \]
    so $q_1(x, w, u, \dot{u}) = x + xw$ and $q_2(x, w, u, \dot{u}) = xu' + xu + w^2$.
    This quadratization is optimal and monomial.
\end{example}

The next theorem addresses the existence of a quadratization for non-autonomous polynomial nonlinear ODEs. 

\begin{theorem}\label{thm:existence_inputs}
   Every polynomial ODE system of the form~\eqref{eq:sys_main_input} has a monomial quadratization assuming that the inputs are differentiable.
   Moreover, the order of the optimal monomial quadratization does not exceed $\Pi_{i=1}^{\nn+r}(d_i + 1)$, where 
   \[
    d_i := \deg_{x_i} \bp(\bx, \bu) \;\; \text{for} \;\; 1 \leqslant i \leqslant \nn \quad\text{ and }\quad d_{\nn + i} := \deg_{u_i} \bp(\bx, \bu), \;\; \text{for} \;\; 1 \leqslant i \leqslant r,
   \]
   are the maximal degrees of the polynomials in the state and inputs, respectively.

\end{theorem}
\begin{proof}
   We consider the system~\eqref{eq:sys_main_input} and construct a quadratization for it.
   We consider the set of monomials
        \begin{align*}
            \cM := \{m(\bx, \bu)   \mid & m(\bx, \bu) \text{ is a monomial in $\bx, \bu$ s.t. } \\
            & \forall\; 1 \leqslant i \leqslant \nn\colon \deg_{x_i} m(\bx, \bu) \leqslant d_i \text{ and }
            \forall\;1 \leqslant i \leqslant r\colon \deg_{u_i} m(\bx, \bu) \leqslant d_{\nn + i}\}
        \end{align*}
   and claim that $\cM$ is a quadratization of the original system~\eqref{eq:sys_main_input}.
   To prove this, we consider an arbitrary $m(\bx, \bu) \in \cM$ and examine the monomials on the right-hand side of its derivative $\dot{m}(\bx, \bu)$, where
   \[
     \dot{m}(\bx, \bu) = \sum\limits_{i = 1}^\nn p_i(\bx, \bu) \frac{\partial m(\bx, \bu)}{\partial x_i} + \sum\limits_{j = 1}^r \dot{u}_j \frac{\partial m(\bx, \bu)}{\partial u_i},
   \]
   which follows by the chain rule. 
   By construction of $\cM$, every monomial in $p_i(\bx, \bu)$ for $1 \leqslant i \leqslant \nn$ belongs to $\cM$, so every monomial in the first sum can be written as $m_0(\bx, \bu)\frac{m(\bx, \bu)}{x_i}$, where $x_i$ appears in $m$ and $m_0 \in \cM$.
   Since $\frac{m}{x_i} \in \cM$, this product is a quadratic expression in $\cM$.
   Furthermore, every monomial in the second sum is of the form $\dot{u}_i \frac{m(\bx, \bu)}{u_i}$, where $u_i$ appears in $m$.
   Since $\frac{m(\bx, \bu)}{u_i} \in \cM$, this monomial can also be written as a quadratic expression in $\cM$.
   The claim about the size of the optimal quadratization follows from the fact that $|\cM| = \Pi_{i=1}^{\nn+r}(d_i + 1)$.
\end{proof}

The external inputs $\bu(t)$ may not always be differentiable, as step inputs, ramp inputs, pulses, and many other non-differentiable inputs are used in engineering. While those can often be smoothed, the next section considers the case where the inputs are not differentiable. 

%%%%%%%%%%%%%%%%%%%%%%%%%%%%%%%%%%%%%%%%%%%%%%%%%%
\subsection{Systems with non-differentiable inputs}\label{sec:input-free}
We consider systems with non-differentiable inputs $\bu(t)$ for which we define a new quadratization. 
While it is not always possible to relax the differentiability condition on the inputs in \cref{thm:existence_inputs} (see \cref{ex:infinite} and \cref{sec:gems} for examples), we characterize the cases when it is possible. We start with a corresponding definition.
\begin{definition}[Input-free quadratization of non-autonomous polynomial ODEs]\label{def:quadr_input_zero}
  In the notation of \cref{def:quadr_input}, an $\ell$-dimensional vector of new variables 
  \begin{equation}
      \bw = \bw(\bx) \quad \in \ \mathbb{C}[\bx]^\ell
  \end{equation}
  is said to be an \emph{input-free quadratization} of~\eqref{eq:sys_main_input} 
  if there exist vectors $\bq_1(\bx, \bw, \bu)$ and $\bq_2(\bx, \bw, \bu)$ of dimensions $\nn$ and $\ell$, respectively, with the entries being polynomials of total degree at most two such that 
  \begin{equation}\label{eq:input_free_result}
  \dot{\bx} = \bq_1(\bx, \bw, \bu) \quad \text{ and }\quad \dot{\bw} = \bq_2(\bx, \bw, \bu).
  \end{equation}
  The number $\ell$ is called \emph{the order of quadratization}.
  A quadratization of the smallest possible order is called \emph{an optimal quadratization}.
\end{definition}

\begin{remark}
  The key difference between \cref{def:quadr_input} and \cref{def:quadr_input_zero} is that $\dot{\bu}$ does not appear in the right-hand side terms of the resulting system (that is, in the vectors $\bq_1$ and $\bq_2$ of quadratic polynomials).
  This, in turn, implies that $\bw$ cannot depend on $\bu$.
  For example, the quadratization from~\Cref{ex:input} is not input-free (and in fact the model does not admit an input-free quadratization, see also~\Cref{ex:infinite}).
\end{remark}

\begin{example}
    Consider a two-dimensional system describing the Duffing oscillator:
    \[
    \dot{x}_1 = x_2, \quad \dot{x}_2 = -\alpha x_1 - \delta x_2 - \beta x_1^3 + u.
    \]
    We claim that $w(\bx) = x_1^2$ is an input-free (monomial and optimal) quadratization.
    Indeed, the corresponding quadratic system is
    \[
    \dot{x}_1 = x_2, \quad \dot{x}_2 = -\alpha x_1 - \delta x_2 - \beta x_1w + u, \quad \dot{w} = 2x_1x_2,
    \]
    so $\bq_1(\bx, w, u) = [x_2, -\alpha x_1 - \delta x_2 - \beta x_1w + u]^\top$ and $q_2(\bx, w, u) = 2x_1x_2$.
    The quadratic system above does not involve any derivatives of $u$ and, thus, does not add any additional constraints on the input function.
\end{example}

One nice feature of input-free quadratizations is that they produce the quadratic-bilinear structure as the following lemma shows.

\begin{lemma}
  Assume that the original system~\eqref{eq:sys_main_input} is input-affine, that is, of the form
  \begin{equation}\label{eq:inp_affine}
  \dot{\bx} = \bp_0(\bx) + \sum\limits_{i = 1}^r \bp_i(\bx) u_i.
\end{equation}
If there exists an input-free quadratization~\eqref{eq:input_free_result} with variables $\bw$, then this quadratization is quadratic-bilinear. 
\end{lemma}
\begin{proof}
   Since every $\dot{x}_i$ is input-affine, the same is true for every $\dot{w}_j$.
   Since the inputs do not appear in the new variables and $w_1, \ldots, w_\ell$ is a quadratization of the original system, the system~\eqref{eq:input_free_result} has at most a quadratic right-hand side which does not involve products of inputs.
   This shows that it is quadratic-bilinear.
\end{proof}

Unlike general quadratizations of non-autonomous systems, an input-affine quadratization may not always exist for a given input-affine polynomial system (see~\cref{ex:infinite}).
In the remainder of this section, we present additional theoretical results on the existence of an input-free quadratization for input-affine systems. 
We start with a simple class of systems admitting an input-free quadratization.

\begin{lemma}\label{lem:input_bilinear}
  Assume that the system~\eqref{eq:inp_affine} is {\rm polynomial-bilinear}, that is, the total degree of $\bp_i(\bx)$ in \eqref{eq:inp_affine} is at most one for every $1 \leqslant i \leqslant r$.
  Then there is an input-free quadratization for~\eqref{eq:inp_affine}.
\end{lemma}

\begin{proof}
  Let $d := \deg \bp_{0}(\bx)$ and assume that the total degree of $\bp_{i}(\bx)$ in \eqref{eq:inp_affine} is at most one for every $1 \leqslant i \leqslant r$. 
  We consider the set of monomials
  \[
    \mathcal{M} = \{m(\bx) \mid m(\bx) \text{ is a monomial in }\bx \text{ and } \deg m(\bx) \leqslant d\},
  \]
  and claim that $\mathcal{M}$ is an input-free quadratization of the polynomial-bilinear system~\eqref{eq:inp_affine}.
  To prove this, consider any $m(\bx) \in \mathcal{M}$, which has derivative
  \[
    \dot{m}(\bx) = \sum\limits_{i = 1}^\nn p_{0, i}(\bx) \frac{\partial m(\bx)}{\partial x_i} + \sum\limits_{j = 1}^r u_j\sum\limits_{i = 1}^\nn p_{j, i}(\bx) \frac{\partial m(\bx)}{\partial x_i}.
  \]
  We separately consider the terms in the first and second sums:
  \begin{enumerate}
      \item In the first sum, every monomial is the form $\frac{m(\bx)}{x_i}m_0(\bx)$, where $m_0(\bx)$ is a monomial of $p_{0, i}(\bx)$. 
      It is product of two elements of $\mathcal{M}$: $\frac{m(\bx)}{x_i}$ and $m_0(\bx)$;
      
      \item In the second sum, every monomial is of the form $\frac{m(\bx)}{x_i}m_0(\bx)u_j$, where $m_0(\bx)$ is a monomial from $p_{i, j}(\bx)$.
      Since $\deg p_{j, i}(\bx) \leqslant 1$, we have $\deg m_0 \leqslant 1$, so $\frac{m(\bx)}{x_i}m_0(\bx) \in \mathcal{M}$.
  \end{enumerate}
  Together, this shows that every monomial in the right-hand side of the equation for $m(\bx)$ is at most quadratic in $\cM$, so $\cM$ is indeed an input-free quadratization.
\end{proof}

More generally, it turns out that the existence of an input-free quadratization can be characterized via the properties of certain linear differential operators associated with the inputs $\bu(t)$.
This characterization turns out to be a generalization of the classical concept of locally finite derivation as explained in~\cref{rem:locally_finite} below. The next proposition provides this existence result.

\begin{proposition}\label{prop:inp_general}
  Consider an input-affine system of the form~\eqref{eq:inp_affine}.
  We introduce $r$ differential operators:
  \[
    D_i := \bp_i(\bx)^\top \cdot \frac{\partial}{\partial \bx},\quad 1 \leqslant i \leqslant r,
  \]
  where $\frac{\partial}{\partial \bx} = \left[ \frac{\partial}{\partial x_1}, \ldots, \frac{\partial}{\partial x_\nn}\right]^\top$.
  Let $\mathcal{A}$ be a subalgebra generated by $D_1, \ldots, D_r$ in the algebra $\mathbb{C}\left[\bx, \frac{\partial}{\partial \bx}\right]$ of all polynomial differential operators in $\bx$.
  Then there is an input-free quadratization for~\eqref{eq:inp_affine} if and only if
  \begin{equation}\label{eq:fin_dim_cond}
      \dim \{ A(x_i) \mid A \in \mathcal{A}\} < \infty \quad \text{for every } 1\leqslant i \leqslant \nn.
  \end{equation}
\end{proposition}

\noindent 
Before proving the proposition, let us illustrate it with two examples.

\begin{example}[Infinite dimension]\label{ex:infinite}
  Consider the scalar input-affine ODE with a single input:
  \[
    \dot{x} = x^2 u.
  \]
  In the notation of~\cref{prop:inp_general}, we have $D_1 = x^2 \frac{\partial}{\partial x}$. Then the algebra $\mathcal{A}$ is spanned by $D_1, D_1^2, \ldots$ and its application to $x$  yields
  \[
    D_1(x) = x^2, \quad D_1^2(x) = 2x^3, \quad D_1^3(x) = 6 x^4, \quad \ldots
  \]
  Thus, $\dim \{A(x) \mid A \in \mathcal{A}\} = \infty$ and the proposition implies that there is no input-free quadratization for this ODE.
\end{example}

\begin{example}[Finite dimension]\label{ex:findim}
  We consider a modification of the previous example:
  \begin{equation}\label{eq:findim}
    \dot{x}_1 = x_1 + x_1 u,\quad \dot{x}_2 = x_1^2 u.
  \end{equation}
  In this case $D_1 = x_1 \frac{\partial}{\partial x_1} + x_1^2 \frac{\partial}{\partial x_2}$.
  We have $D_1(x_1) = x_1$ and $D_1(x_2) = x_1^2$. 
  Therefore, $\{A(x_1) \mid A \in \mathcal{A}\} = \operatorname{span}\{ x_1 \}$ and $\{A(x_2) \mid A \in \mathcal{A}\} = \operatorname{span}\{ x_2, x_1^2\}$. 
  Therefore, \cref{prop:inp_general} implies that~\eqref{eq:findim} admits an input-free quadratization.
\end{example}
The condition in \cref{prop:inp_general} is related to locally finite derivations, as we explain next.

\begin{remark}[Algorithmic decidability of the condition~\eqref{eq:fin_dim_cond}]\label{rem:locally_finite}
  For the special case $r = 1$, the condition~\eqref{eq:fin_dim_cond} is equivalent to saying that the operator $D_1$ is \emph{locally finite}.
  A differential operator $D$ on a polynomial ring in variables $x_1, \ldots, x_\nn$ is called locally finite if the dimension of $\operatorname{span}\{x_i, D(x_i), D^2(x_i), \ldots\}$ is finite for every~$i$.
  Locally finite derivations appear in different contexts~\cite{locally_finite} but the problem of determining algorithmically if a given derivation is locally finite has been solved only for the bivariate case~\cite[Section~4]{locally_finite} and remains an open problem in the general case.
\end{remark}

\begin{proof}[Proof of proposition~\ref{prop:inp_general}]
  $\Leftarrow$: Assume that $\{ A(x_i) \mid A \in \mathcal{A}\}$ is finite dimensional for every $1 \leqslant i \leqslant \nn$.
  For each of these spaces we choose a basis, and let $\mathcal{M}$ be the union of these bases.
  Consider any $p(\bx) \in \mathcal{M} \cup \{x_1, \ldots, x_\nn\}$.
  Then
  \begin{equation}\label{eq:deriv_p}
      \dot{p}(\bx) = \sum\limits_{j = 1}^\nn p_{0, j}(\bx) \frac{\partial p}{\partial x_i} + \sum\limits_{i = 1}^r u_i D_i(p(\bx)).
  \end{equation}
  Since $p(\bx)$ is of the form $A(x_j)$ for some $A \in \mathcal{A}$ and $1 \leqslant j \leqslant \nn$, $D_i(p(\bx))$ is also of this form. Thus, $D_i(p(\bx))$ can be written as a linear combination of elements of $\mathcal{M} \cup \{x_1, \ldots, x_\nn\}$.
  Therefore, if we collect equations~\eqref{eq:deriv_p} for $p(\bx)$ ranging over $\mathcal{M} \cup \{x_1, \ldots, x_\nn\}$, we obtain an ODE system in the variables $\{x_1, \ldots, x_\nn\} \cup \mathcal{M}$, and this system satisfies the requirements of \cref{lem:input_bilinear}.
  Therefore, this system has an input-free quadratization, and, thus, the same is true for the original system.
  
  $\Rightarrow$: Assume that~\eqref{eq:inp_affine} has an input-free quadratization with the $\ell$-dimensional vector of new variables $\mathbf{w}(\bx)$. Let $\mathcal{X} := \operatorname{span}\{x_1, \ldots, x_\nn, w_1, \ldots, w_\ell\}$.
  We prove by induction on an integer $h \in \mathbb{Z}_{\geqslant 0}$ that, for every $1 \leqslant i_1, \ldots, i_h \leqslant r$ and $1 \leqslant j \leqslant \nn$, the element $[D_{i_1}\ldots D_{i_h}](x_j)$ belongs to $\mathcal{X}$.
  The base case $h = 0$ is true.
  Assume that the statement has been proven for $h$, and consider any $1 \leqslant i_1, \ldots, i_h \leqslant r$ and $1 \leqslant j \leqslant \nn$ and let $p(\bx) = [D_{i_1}\ldots D_{i_h}](x_j)$.
  Since $p(\bx) \in \mathcal{X}$ by the induction hypothesis, $\dot{p}(\bx)$ can be written as at most quadratic polynomial in $\mathcal{X}$ and $\bu$.
  From the expression~\eqref{eq:deriv_p} one can observe that this implies that, for every $1 \leqslant s \leqslant r$, $D_s(p(\bx)) = [D_sD_{i_1}\ldots D_{i_h}](x_j) \in \mathcal{X}$.
  Since this holds for any $1 \leqslant i_1, \ldots, i_h \leqslant r$ and $1 \leqslant j \leqslant \nn$, this proves the induction step.
  Since $\mathcal{X}$ is finite-dimensional, the dimension of $\{A(x_i) \mid A \in \mathcal{A}\}$ must be finite for every $1 \leqslant i \leqslant \nn$.
\end{proof}

Although it may be complicated to verify the general condition~\eqref{eq:fin_dim_cond} from~\cref{prop:inp_general}, there is an important special case when it can be easily verified.

\begin{lemma}\label{lem:triangular}
   In the notation of \cref{prop:inp_general}, if, for every $1 \leqslant i \leqslant r$ and $1 \leqslant j \leqslant \nn$,
   \[
     D_i(x_j) \in \mathbb{C}[x_{1}, \ldots, x_{j - 1}]
   \]
   holds, then the condition $\dim \{A(x_i) \mid A \in \mathcal{A}\} < \infty$ is fulfilled.
\end{lemma}

\begin{proof}
  The assumption that, for every $1 \leqslant i \leqslant r$ and $1 \leqslant j \leqslant \nn$, $D_i(x_j)$ depends only on $x_k$ with $k > j$ implies that, for every $1 \leqslant i_1, \ldots, i_\nn \leqslant \nn$, the product $D_{i_1}\ldots D_{i_\nn}$ sends each of $x_1, \ldots, x_\nn$ to zero.
  Therefore, the dimension of each $\{A(x_i) \mid A \in \mathcal{A}\}$ does not exceed $1 + r + \ldots + r^{\nn - 1}$.
\end{proof}

%%%%%%%%%%%%%%%%%%%%%%%%%%%%%%%%%%%%%%
\section{Dimension-agnostic quadratizations of families of ODE systems}\label{sec:agnostic}
This section presents a new quadratization method that is applicable to families of ODE systems of variable dimension, for which it produces a dimension-agnostic quadratization. The most natural use-case is ODEs that are derived via semi-discretization of PDEs, i.e., where the symbolic form of the nonlinear terms stays the same but the discretization dimension $\nx$ of the system can be varied. 
The current way to do this is to compute a quadratization for every $\nx$ separately. 
However, this is time-consuming for large $\nx$ and appears unnecessary. The methods derived in this section improve upon this by exploiting similar structures in these systems.
We note that while semi-discretized PDEs are presented in the first example, the methods herein apply to a wider class of problems than that, as we see below. 
To motivate the setup and subsequent definitions, let us start with an example.

\begin{example} \label{ex:PDE_quadratization}
Consider a special case of a dimensionless scalar PDE used in traffic flow modeling~\cite[\S 12.6]{haberman2019applied}:
\[
  \frac{\partial \rho}{\partial t}(t, \xi) = \rho(t, \xi) + \rho^2(t, \xi) \frac{\partial \rho(t, \xi)}{\partial \xi}, \qquad \rho(t, 0) = 0, \quad \rho(t, 1) = 1,\; 
\]
where $\rho(t,\xi)$ is the traffic density, $t$ denotes times, and $\xi \in [0, 1]$ is the spatial variable. Let the initial condition be a square-integrable function $\rho(0, \xi) = \varphi(\xi)$. We discretize the spatial domain uniformly and denote the density at the nodal values as $x_{i}^{[\nx]}(t) = \rho(t, i/(\nx + 1))$  for $1 \leq i \leq \nx$. 
This defines the $\nx$-dimensional state vector $\bx^{[\nx]}(t) = [x_{1}^{[\nx]}(t), \ldots, x_{\nx}^{[\nx]}(t)]^\top$. 
We approximate the spatial derivative by a backward first-order differencing formula as
\[
  \frac{\partial \rho}{\partial \xi}(t, i/(\nx + 1)) \approx \frac{\rho(t, i/(\nx+1)) - \rho(t, (i - 1)/(\nx+1))}{\Delta \xi},
\]
where $\Delta \xi = 1/(\nx+1)$, which then yields the $\nx$-dimensional system of ODEs
\begin{equation}\label{eq:ex_discr}
  \dot{\bx}^{[\nx]} = \bx^{[\nx]} + (\bx^{[\nx]})^2 \odot (\bD\bx^{[\nx]}),
\end{equation}
where $\odot$ denotes the Hadamard (or element-wise) product and where
\begin{equation}\label{eq:ex_Dmatrix}
  \bD = \frac{1}{\Delta \xi} \left[
    \begin{array}{cccccc}
        1 & 0 &  & & &\\
        -1 & 1 & 0 & & & \\
        0&  -1 & 1 & & &\\
        & & \ddots & \ddots & \ddots &  \\
        & & & -1 & 1 & 0 \\
         & & & 0 & -1 & 1 \\
    \end{array}
    \right]\in\mathbb{R}^{\nx \times \nx}.
\end{equation}
Our goal is to quadratize~\eqref{eq:ex_discr} for arbitrary $\nx$; for brevity we drop the superscript $[\nx]$ in what follows.
Consider the vectors $\bw_1 = \bx^2$ and $\bw_2 = \bx \odot (\bS \bx)$, where $\bS$ is the lower shift matrix (with ones on the subdiagonal and zeros elsewhere), i.e. $\bS\bx = [0, \; x_1, \;\ldots, \; x_{\nx - 1}]^\top$.
We claim that the nonzero entries of $\bw_1$ and $\bw_2$ are a quadratization of~\eqref{eq:ex_discr} for any $\nx$.
The original system can be now written as:
\begin{equation}\label{eq:pde_discr}
  \dot{\bx} = \bx + \bw_1 \odot (\bD \bx).
\end{equation}
Furthermore, using $\bD\bx = \frac{1}{\Delta \xi}(\bx - \bS\bx)$ we have:
\begin{align}\label{eq:pde_ex_quadratized}
\begin{split}
  \dot{\bw}_1 &= 2\bx(\bx + \bx^2 \odot (\bD\bx)) = 2\bw_1 + \frac{2}{\Delta \xi}\bw_1(\bw_1 - \bw_2),\\
  \dot{\bw}_2 &= (\bx + \bx^2 \odot \frac{1}{\Delta \xi}(\bx - \bS\bx))(\bS \bx) + \bx (\bS \bx + (\bS\bx)^2\odot \frac{1}{\Delta \xi}(\bS \bx - \bS^2 \bx)) = \\
  &= 2\bw_2 + \frac{1}{\Delta \xi}(\bw_1\odot \bw_2 - \bw_2^2 + \bw_2\odot(\bS \bw_1) - \bw_2\odot(\bS \bw_2)) 
\end{split}
\end{align}
Therefore, $\bw_1 = \bw^{[\nx]}_1$ and $\bw_2 = \bw^{[\nx]}_2$ is a quadratization of~\eqref{eq:ex_discr} for every $\nx$, which we refer to as a \emph{dimension-agnostic} quadratization (see~\cref{def:da} for a formal definition). 
Note, that here we can solve for $\bw_1$ and $\bw_2$ and then recover $\bx = \sqrt{\bw_1}$, so the state equations are decoupled and we do not need to solve the equation for $\dot{\bx}$. 
\end{example}

We highlight that there are two types of quadratizing variables in the above example: \emph{uncoupled} variables $\bw_1$, where each variable depends only on one of the $x_i$'s, and \emph{coupled} variables $\bw_2$, where each variable involves two different $x_i$ and $x_j$ so that the $i$th and $j$th equations are coupled, which is the case when $|i - j| = 1$.

We next provide a formal definition for the class of systems we consider. If those are derived from PDEs, then $\nd$ is the number of dependent variables in the polynomial PDE system. For example, $\nd = 1$ for \Cref{ex:PDE_quadratization} and $\nd = 2$ for the solar wind example~\eqref{eq:hux_poly} in \cref{sec:heliospheric} which was lifted with one additional variable to polynomial form. Moreover, $\nx$ is the dimension of spatial discretization. 

\begin{definition}[Family of linearly coupled polynomial ODEs]\label{def:family}
    Let $\bp_0(\bx), \ldots, \bp_\nd(\bx) \in \mathbb{C}[\bx]^\nd$ be $\nd + 1$ vectors with entries polynomial in $\bx = [x_1, \ldots, x_\nd]^\top$.
    For such $\bp_0(\bx), \ldots, \bp_\nd(\bx)$ we assign \emph{a family of linearly coupled ODE systems} as follows.
    Consider a positive integer $\nx$ and matrices $\bD_1, \ldots, \bD_\nd \in \mathbb{C}^{\nx \times \nx}$.
    We construct an ODE system for $\nn = \nd \cdot \nx$ variables $x^{[\nx]}_{i, j}$ with $1 \leqslant i \leqslant \nx$ and $1 \leqslant j \leqslant \nd$.
    A vector of these variables is denoted by $\bx^{[\nx]}$, and we define
$\bx^{[\nx]}_{i}=[x^{[\nx]}_{i, 1}, \ldots, x^{[\nx]}_{i, \nd}]^\top$ and $\bx^{[\nx]}_{*, j} = [x^{[\nx]}_{1, j}, \ldots, x^{[\nx]}_{\nx, j}]^\top$.
    Then the ODE system defined by the polynomial vectors $\bp_0(\bx), \ldots, \bp_\nd(\bx)$ and constant matrices $\bD_1, \ldots, \bD_\nd$ is
    \begin{equation}\label{eq:from_template}
    \dot{\bx}^{[\nx]}_{i} = \bp_0(\bx^{[\nx]}_{i}) + \sum\limits_{j = 1}^\nd \bp_j(\bx^{[\nx]}_{i}) (\bD_j\bx^{[\nx]}_{*, j})_i, \quad i=1,2,\ldots, \nx
    \end{equation}
    We refer to the ODE system~\eqref{eq:from_template} as $\cF^{[\nx]}(\bp_0, \ldots, \bp_\nx, \bD_1, \ldots, \bD_\nd)$.
\end{definition}

\begin{example}[Vectors $\bp_0, \ldots, \bp_\nd$ for~\eqref{eq:ex_discr}]
    Recall the system \eqref{eq:ex_discr} where $\nd = 1$. We can rewrite this system into the form of equation~\eqref{eq:from_template}  with $\bp_0(x) = [x]$ and $\bp_1(x) = [x^2]$ and as well as the matrix $\bD_1=\bD$ from~\eqref{eq:ex_Dmatrix}.
\end{example}

Since the infinite family of ODE systems~\eqref{eq:from_template} is defined by a finite number of variables, that is, by $\nd + 1$ polynomial vectors $\bp_0(\bx), \ldots, \bp_\nd(\bx)$, it is natural to define an entire family of quadratizations also by some finite amount of variables.

\begin{definition}[Dimension-agnostic quadratization]\label{def:da}
    Consider a family of linearly coupled ODEs defined by $\nd + 1$ polynomial vectors $\bp_0(\bx), \ldots, \bp_\nd(\bx)$ in $\bx = [x_1, \ldots, x_\nd]^\top$, and consider an $\ell$-dimensional vector $\bw_1(\bx) \in \mathbb{C}[\bx]^\ell$ together with an $L$-dimensional vector $\bw_2(\bx, \tilde{\bx}) \in \mathbb{C}[\bx, \tilde{\bx}]^L$, where $\tilde{\bx} = [\tilde{x}_1, \ldots, \tilde{x}_\nd]^\top$ are formal variables used as placeholders for the coupled variables (to be made precise below). 
     For every integer $\nx$ and matrices $\bD_1, \ldots, \bD_\nd \in \mathbb{C}^{\nx \times \nx}$, we define a set 
    \begin{multline}\label{eq:Magnostic}
      \cM(\bw_1, \bw_2; \bD_1, \ldots, \bD_\nd) := \{ \bw_{1}(\bx^{[\nx]}_{i}) \mid 1 \leqslant i \leqslant \nx\} \;\cup \\  \{ \bw_{2}(\bx^{[\nx]}_{ i_0}, \bx^{[\nx]}_{i_1}) \mid 1 \leqslant i_0 \neq i_1 \leqslant \nx, \text{ and }\exists k\colon (\bD_k)_{i_0, i_1} \neq 0 \}.
    \end{multline}
    Then we say that $\bw_1$ and $\bw_2$ are a \emph{dimension-agnostic quadratization} of the family if, for every integer $\nx$ and matrices $\bD_1, \ldots, \bD_\nd\in \mathbb{C}^{\nx \times \nx}$, $\cM(\bw_1, \bw_2; \bD_1, \ldots, \bD_\nd)$
    is a quadratization of $\cF^{[\nx]}(\bp_0, \ldots, \bp_\nd, \bD_1, \ldots, \bD_\nd)$.
\end{definition}

\begin{example}[Vectors $\bw_1$ and $\bw_2$ for~\eqref{eq:ex_discr}]\label{ex:da_quadr_basic}
    Recall that in \cref{ex:PDE_quadratization} we obtained a quadratization with the vectors
    \[
      \bw_1^{[\nx]} = [(x_1^{[\nx]})^2, \ldots, (x_\nx^{[\nx]})^2]^\top \quad \text{and} \quad \bw_2^{[\nx]} = [0, x_1^{[\nx]}x_2^{[\nx]}, x_2^{[\nx]}x_3^{[\nx]}, \ldots, x_{\nx - 1}^{[\nx]}x_\nx^{[\nx]}]^\top.
    \]
    We claim that  these new variables are $\cM(\bw_1, \bw_2; \bD)$, where $\bw_1 = x^2$ and $\bw_2 = x\cdot \tilde{x}$ in this case  are scalars.
    Indeed, $\bw_1 = x^2$ gives rise to $\bw_1^{[\nx]}$ by~\eqref{eq:Magnostic}.
    Since the off-diagonal nonzero entries of $\bD$~\eqref{eq:ex_Dmatrix} are at the $(i, i - 1)$ locations for $2 \leqslant i \leqslant \nx$, so the new variables in~\eqref{eq:Magnostic} coming from $\bw_2 = x\cdot \tilde{x}$ are obtained by setting $x = x_i^{[\nx]}$ and $\tilde{x} = x_{i - 1}^{[\nx]}$ for every $2 \leqslant i \leqslant \nx$ yielding  the nonzero entries of $\bw_2^{[\nx]}$.
    Assume, for example, that we had periodic boundary conditions and, thus (using $\Delta \xi = \frac{1}{n}$)
    \[
    \bD_{\operatorname{per}} = \frac{1}{\Delta \xi} \left[
    \begin{array}{cccccc}
        1 & 0 &  & & & -1\\
        -1 & 1 & 0 & & & \\
        0&  -1 & 1 & & &\\
        & & \ddots & \ddots & \ddots &  \\
        & & & -1 & 1 & 0 \\
         & & & 0 & -1 & 1 \\
    \end{array}
    \right]\in\mathbb{R}^{\nx \times \nx},
    \]
    which has an additional nonzero off-diagonal entry (the $(1, \nx)$th entry).
    This results in
    \[
      \cM(\bw_1, \bw_2; \bD_{\operatorname{per}}) = \cM(\bw_1, \bw_2; \bD) \cup \{ x_1^{[\nx]}x_{\nx}^{[\nx]}\}.
    \]
    Furthermore, it is possible to show that $\cM(\bw_1, \bw_2; \bD)$ is a quadratization of~\eqref{eq:ex_discr} for any complex-valued matrix $\bD$ (see~\Cref{ex:use_agnostic}).
\end{example}

Although the requirements on the dimension-agnostic quadratization---that is, having the same ``shape'' of new variables for all dimensions and matrices---may appear restrictive, we show that such a quadratization always exists.

\begin{theorem}\label{thm:existencEgnostic}
  Every family of linearly coupled ODEs admits a monomial dimension-agnostic quadratization; this can be chosen such that $\deg_{\tilde{\bx}} \bw_{2}(\bx, \tilde{\bx}) = 1$.
  
  Furthermore, if the family is defined by $\nd + 1$ polynomial vectors $\bp_0(\bx), \ldots,\bp_\nd(\bx)$ with $d := \max\limits_{0 \leqslant j \leqslant \nd}\deg \bp_j(\bx)$, then there exists a dimension-agnostic quadratization with $\ell \leqslant \binom{\nd + d}{d}$ and $L \leqslant \nd\binom{\nd + d}{d}$ (in the notation of~\cref{def:da}).
\end{theorem}

\begin{proof}
  We construct such a quadratization explicitly.
  Let $w_{1, 1}(\bx), \ldots, w_{1, \ell}(\bx)$ be the set of all monomials of degree at most $d$ in $\bx$ and let $w_{2, 1}(\bx, \tilde{\bx}), \ldots, w_{2, L}(\bx, \tilde{\bx})$ be the set of all $2n$-variate monomials of degree at most one in $\tilde{\bx}$ degree at most $d$ in $\bx$.
  Then $\ell = \binom{\nd + d}{d}$ and $L = \nd\binom{\nd + d}{d}$.
  We next show that these $\bw_1 = [w_{1, 1}(\bx), \ldots, w_{1, \ell}(\bx)]^\top$ and $\bw_2 = [w_{2, 1}(\bx, \tilde{\bx}), \ldots, w_{2, L}(\bx, \tilde{\bx})]^\top$ provide a dimension-agnostic quadratization for the family defined by $\bp_0(\bx), \ldots,\bp_\nd(\bx)$.
  
  In order to prove this, consider arbitrary $\nx$ and matrices $\bD_1, \ldots, \bD_\nd \in \mathbb{C}^{\nx \times \nx}$.
  We consider the corresponding set $\cM := \cM(\bw_1, \bw_2; \bD_1, \ldots, \bD_\nd)$ as in~\eqref{eq:Magnostic}.
  First we observe that every monomial in the right-hand side of $\cF^{[\nx]}(\bp_0, \ldots,\bp_\nd, \bD_1, \ldots, \bD_\nd)$ (see~\eqref{eq:from_template}) belongs to $\cM$.
  Indeed, consider $\dot{\bx}_{i}^{[\nx]}$ for some $1 \leqslant i \leqslant \nx$. 
  If a monomial in $\dot{\bx}_{i}^{[\nx]}$ comes from the $\bp_0(\bx_{i}^{[\nx]})$, then it is the form $m(\bx_{i}^{[\nx]})$ with $\deg m \leqslant d$ and, thus, can be written as $w_{1, k}(\bx_{i}^{[\nx]})$ for some $1 \leqslant k \leqslant \ell$.
  Otherwise, it is of the form $m(\bx_{i}^{[\nx]}) x_{i_0, j}$ for some $1 \leqslant j \leqslant \nd$ and $1 \leqslant i_0 \leqslant \nx$ such that $\dot{\bx}^{[\nx]}_i$ depends on $\bx_{i_0}^{[\nx]}$ and $\deg m \leqslant d$, and therefore can be presented as $w_{2, k}(\bx_{i}^{[\nx]}, \bx_{i_0}^{[\nx]})$ for some $1 \leqslant k \leqslant L$.
  Now we consider any element of $m(\bx^{[\nx]}) \in \cM$.
  It belongs to either $\bw_1(\bx^{[\nx]}_{i})$ for some $i$ or to $\bw_2(\bx^{[\nx]}_{i_0}, \bx^{[\nx]}_{i_1})$ for some $i_0, i_1$.
  Consider the case of $\bw_1(\bx^{[\nx]}_{i})$:
  \[
    \dot{\bw}_1(\bx^{[\nx]}_{i}) = \sum\limits_{k = 1}^\nd \dot{x}_{i, k}^{[\nx]} \frac{\partial \bw_{1}(\bx^{[\nx]}_{i})}{\partial x^{[\nx]}_{i, k}}.
  \]
  Since every monomial of $\dot{x}^{[\nx]}_{i, k}$ belongs to $\cM$ and every monomial in $\frac{\partial \bw_{1}(\bx^{[\nx]}_{i})}{\partial x^{[\nx]}_{i, k}}$ belongs to $\cM$, the right-hand side of the equation above is quadratic in $\cM$.
  The case of $m(\bx^{[\nx]})$ belonging to $\bw_{2}(\bx^{[\nx]}_{i_0}, \bx^{[\nx]}_{i_1})$ is completely analogous.
\end{proof}

While this existence result is encouraging and yields a dimension-agnostic quadratization, the quadratization from the proof of \cref{thm:existencEgnostic} is large.
To find a more compact quadratization, one would need to check if arbitrary $\bw_1(\bx)$ and $\bw_2(\bx, \tilde{\bx})$ yield a dimension-agnostic quadratization.
Checking this directly from~\cref{def:da} is not possible since the definition involves a universal quantifier.
The next proposition gives a simple way of checking this.
It turns out that it is sufficient to consider the case $n = 4$ with particular matrices $\bD_i$.
Note that, although the matrices $\bD_i$ corresponding to a particular difference scheme of interest will be likely different from the ones given in the proposition below, it is guaranteed that a dimension-agnostic quadratization which works for the matrices in the proposition will work for any other matrices as well (for the way this works in practice, see~\Cref{sec:alg_agnostic}).
The intuition behind this is that $n = 4$ is sufficient to exhibit all important combinatorial types of coupling, and, thus, a general matrix will always look locally as one of the fragments of the $4\times 4$ matrix from the proposition. The proof of the next proposition formalizes this.

\begin{proposition}\label{prop:agnostic}
  Consider a family of linearly coupled ODEs defined by $\nd + 1$ polynomial vectors $\bp_0(\bx), \ldots,\bp_\nd(\bx)$.
  Then polynomial vectors $\bw_1(\bx)$ and $\bw_2(\bx, \tilde{\bx})$ yield a monomial dimension-agnostic quadratization of the family if and only if $\cM(\bw_1, \bw_2; \bD_1^*, \ldots, \bD_\nd^*)$ from equation~\eqref{eq:Magnostic} is a quadratization for $\cF^{[\nx]}(\bp_0, \ldots,\bp_\nd,\bD_1^*, \ldots, \bD_\nd^*)$, where $\nx = 4$ and $\bD_1^*, \ldots, \bD_\nd^*$ are defined as follows:
  \[
  \bD_i^* = \begin{pmatrix}
    a_i & b_i & 0 & c_i \\ 
    0 & d_i & e_i & 0\\ 
    0 & 0 & f_i & 0\\
    0 & 0 & 0 & g_i
  \end{pmatrix},
  \]
  where $1 \leqslant i \leqslant \nd$ and $a_i, b_i, c_i, d_i, e_i, f_i, g_i$ are scalar parameters.
\end{proposition}

\begin{proof}
  Assume that $\cM_0 := \cM(\bw_1, \bw_2; \bD_1^*, \ldots, \bD_\nd^*)$ is a quadratization for $\cF_4 := \cF^{[4]}(\bp_0, \ldots,\bp_\nd,\bD_1^*, \ldots, \bD_\nd^*)$ and consider arbitrary $\nx$ and matrices $\bD_1, \ldots, \bD_\nd \in \mathbb{C}^{\nx \times \nx}$.
  We set $\cM_1 := \cM(\bw_1, \bw_2; \bD_1, \ldots, \bD_\nd)$.
  We show that $\cM_1$ is a quadratization for $\cF_\nx := \cF^{[\nx]}(\bp_0, \ldots,\bp_\nd,\bD_1, \ldots, \bD_\nd)$.
  Consider $1 \leqslant i_0 \neq i_1 \leqslant \nx$ such that $(\bD_k)_{j_0, j_1} \neq 0$ for some $k$.
  We show that $\dot{\bw}_{2}(\bx^{[\nx]}_{i_0}, \bx^{[\nx]}_{i_1})$ is quadratic in $\cM_1$.
  For every monomial $m(\bx^{[\nx]})$ in $\dot{\bw}_{2}(\bx^{[\nx]}_{i_0}, \bx^{[\nx]}_{i_1})$, there are the two options:
  \begin{enumerate}
      \item $m$ depends only on $\bx^{[\nx]}_{i_0}$ and $\bx^{[\nx]}_{i_1}$.
      Then $m(\bx^{[4]}_{1}, \bx^{[4]}_{2})$ is present in $\cF_4$ on the right-hand side of $\dot{\bw}_{2}(\bx^{[4]}_{1}, \bx^{[4]}_{2})$.
      Hence, $m(\bx^{[4]}_{1}, \bx^{[4]}_{2})$ is quadratic in $\bw_1(\bx^{[4]}_{1}), \bw_1(\bx^{[4]}_{2})$ and $\bw_2(\bx^{[4]}_{1}, \bx^{[4]}_{2})$.
      Since $\bw_1(\bx^{[\nx]}_{i_0}), \bw_1(\bx^{[\nx]}_{i_1}), \bw_2(\bx^{[\nx]}_{i_0}, \bx^{[\nx]}_{i_1}) \in \cM_1$, we have that $m(\bx^{[\nx]}_{i_0}, \bx^{[\nx]}_{i_1})$ is quadratic in $\cM_1$ as well.
      
      \item $m$ depends on $\bx^{[\nx]}_{i_0}, \bx^{[\nx]}_{i_1}$, and $\bx^{[\nx]}_{i_2}$ for some $i_2 \neq i_0, i_1$.
      \begin{itemize}
          \item If $m$ comes from $\frac{\partial \bw_{2}(\bx^{[\nx]}_{i_0}, \bx^{[\nx]}_{i_1})}{\partial x^{[\nx]}_{i_0, s}} \dot{x}^{[\nx]}_{i_0, s}$ for some $s$, we have that $(\bD_k)_{i_0, i_2} \neq 0$ for some $k$.
          Then $m(\bx^{[4]}_{1}, \bx^{[4]}_{2}, \bx^{[4]}_{4})$ appears in $\cF_4$ in the right-hand side of $\dot{\bw}_{2}(\bx^{[4]}_{1}, \bx^{[4]}_{2})$. 
          Hence, $m(\bx^{[4]}_{1}, \bx^{[4]}_{2}, \bx^{[4]}_{4})$ is quadratic in 
          \[
          \bw_1(\bx^{[4]}_{1}),\; \bw_1(\bx^{[4]}_{2}),\; \bw_1(\bx^{[4]}_{4}),\; \bw_2(\bx^{[4]}_{1},\; \bx^{[4]}_{2}),\; \bw_2(\bx^{[4]}_{1},\; \bx^{[4]}_{4}).
          \]
          Getting back to $\cF_\nx$, since
          \[
          \bw_1(\bx^{[\nx]}_{i_0}),\; \bw_1(\bx^{[\nx]}_{i_1}),\; \bw_1(\bx^{[\nx]}_{i_2}),\; \bw_2(\bx^{[\nx]}_{i_0},\; \bx^{[\nx]}_{i_1}),\; \bw_2(\bx^{[\nx]}_{i_0},\; \bx^{[\nx]}_{i_2}) \in \cM_1,
          \]
          we have that $m(\bx^{[\nx]}_{i_0}, \bx^{[\nx]}_{i_1}, \bx^{[\nx]}_{i_2})$ is quadratic in $\cM_1$.
          \item If $m$ comes from $\frac{\partial \bw_{2}(\bx^{[\nx]}_{i_0}, \bx^{[\nx]}_{i_1})}{\partial x^{[\nx]}_{j_1, s}} \dot{x}^{[\nx]}_{i_1, s}$ for some $s$, we have that $(\bD_k)_{i_1, i_2} \neq 0$ for some $k$.
          The argument is the same as in the previous case but using $m(\bx^{[4]}_{1}, \bx^{[4]}_{2}, \bx_{3})$ instead of $m(\bx^{[4]}_{1}, \bx^{[4]}_{2}, \bx^{[4]}_{4})$.
      \end{itemize}
  \end{enumerate}
  The proof that $\dot{\bx}^{[\nx]}$ and $\dot{\bw}_{1}(\bx^{[\nx]}_{i})$ are quadratic in $\cM_1$ is analogous but simpler because only the first case (i.e., $m$ depends on $\bx^{[\nx]}_{i_0}$ and $\bx^{[\nx]}_{i_1}$ for some $i_0, i_1$) is possible.
  Thus, $\cM_1$ is a quadratization for $\cF_\nx$.
  Since $N$ and $\bD_1, \ldots, \bD_\nx$ were chosen arbitrarily, $\bw_1$ and $\bw_2$ yield a dimension-agnostic quadratization for the family defined by $\bp_0, \ldots,\bp_\nd$.
\end{proof}

%%%%%%%%%%%%%%%%%%%%%%%%%%%%%%%%%%%%%%%%%%%%%%%%%%%%%%%%%%%%%%%%%%%%%%%%%%%%%
\section{The \Qbee algorithm and software for quadratization} \label{sec:Qbee}
%%%%%%%%%%%%%%%%%%%%%%%%%%%%%%%%%%%%%%%%%%%%%%%%%%%%%%%%%%%%%%%%%%%%%%%%%%%%%
To the best of our knowledge, there are two software tools that generate quadratizations of different kind: \textsf{Biocham}~\cite{Biocham} and \Qbee~\cite{QBee}. We focus on \Qbee \ because of its performance, stricter optimality guarantees~\cite[Table~3]{Bychkov2021}, and its flexible algorithm design. In this work, we extend the capability of \Qbee according to the theory from the preceding sections. 
A Jupyter notebook with the examples from~\crefrange{sec:Qbee}{sec:heliospheric} is available online\footnote{Jupyter notebook: \url{https://github.com/AndreyBychkov/QBee/blob/master/examples/Examples_BOPK2023.ipynb}}

%%%%%%%%%%%%%%%%%%%%%%%%%%%%%%%%%%%%%%%%%%%%%%%%%%%%%%%%
\subsection{Review of the original \Qbee algorithm}
The \Qbee algorithm from~\cite{QBee} takes as an input a system of polynomial ODEs and produces as an output an optimal monomial quadratization (in the sense of \cref{def:monomial_quadr}). The \Qbee software is implemented in Python. 
At the code level, \Qbee can be loaded by \mintinline{python}{import sympy as sp; from qbee import *}.
After that, for example, the system 
\[
  \dot{x}_1 = x_1^3 + x_2^2, \quad \dot{x}_2 = x_1 + x_2
\]
can be quadratized by 
\begin{minted}[frame=lines]{python}
x1, x2 = functions("x1, x2")
system = [ # pairs of the form (x, x')
    (x1, x1**3 + x2**2),
    (x2, x1 + x2)
]
quadratize(system).print()
\end{minted}
where $\texttt{system}$ defines a polynomial ODE system $\dot{\bx} = \bp(\bx)$ as a list of pairs $(x_i, p_i(\bx))$. \Qbee prints the quadratic symbolic form of the system and a list of auxiliary variables needed to do that.
The list is guaranteed to be as short as possible.
For the code above, \Qbee produces
\begin{minted}[frame=lines]{text}
    Introduced variables:
    w0 = x1**2
    w1 = x2**2

    x1' = x1*w0 + w1
    x2' = x1 + x2
    w0' = 2*x1*w1 + 2*w0**2
    w1' = 2*x1*x2 + 2*w1
\end{minted}

The algorithms implemented in \Qbee are described in detail in~\cite{Bychkov2021}. To keep this paper self contained, we briefly summarize the key ideas of the algorithm, which we built upon later.  The algorithm follows the general branch-and-bound paradigm~\cite{BB} and the computation of a quadratization is organized as a tree of recursive iterations. Each recursive iteration takes as input the original system, current set $\mathcal{S}$ of new variables ($\mathcal{S} = \varnothing$ at the first iteration), and the smallest order $\ell_0$ of quadratization found so far, where initially we set $\ell_0$ to be the size of quadratization from Theorem~\ref{thm:existence_simple}.
If $\mathcal{S}$ is a quadratization, it is returned.
If $|\mathcal{S}| \geqslant \ell_0$, this means that extending $\mathcal{S}$ does not lead to a better quadratization than already found, this branch is skipped.
Furthermore, to keep the computation manageable, \Qbee implements additional pruning rules that allow the algorithm to decide that the branch can be safely skipped even if $|\mathcal{S}| < \ell_0$~\cite[Section~5]{Bychkov2021}.
Otherwise, the algorithm generates possible useful extensions of $\mathcal{S}$ by looking at not-yet-quadratized monomials in the derivatives of the original variables and elements of $\mathcal{S}$. It runs recursively on these extensions and updates $\ell_0$ if any of these recursive iterations find a better quadratization.
The bound $\ell_0$ forces the tree of recursive calls to be finite and thus implies that the algorithm terminates.
The recursive step of the algorithm is summarized in~\cref{alg:qbee_outline}.

\begin{algorithm}[H]
\caption{Outline of the original \Qbee algorithm from \cite{Bychkov2021}.}
\label{alg:qbee_outline}
\begin{description}[itemsep=0pt]
\item[Inputs]
\begin{itemize} 
\item Symbolic form of the polynomial ODE system $\dot{\bx} = \bp(\bx)$;
    \item A set of potential new variables $\mathcal{S} \subset \mathbb{C}[\bx]$. If not specified, $\mathcal{S} = \varnothing$;
    \item The smallest order $\ell_0$ of quadratization found so far.
    If not specified, we set $\ell_0 = \prod_{i = 1}^\nn(d_i + 1)$, where $d_i = \deg_{x_i}\bp(\bx)$ (see~\cref{thm:existence_simple}).
\end{itemize}
\item[Output] Optimal quadratization of $\dot{\bx} = \bp(\bx)$ extending $\mathcal{S}$ with $|\mathcal{S}| < \ell_0$ or \texttt{nothing} if there is no such quadratization.
\end{description}

\begin{enumerate}[label = \textbf{(Step~\arabic*)}, leftmargin=*, align=left, labelsep=2pt, itemsep=4pt]
    \item If $\mathcal{S}$ is a quadratization and $|\mathcal{S}| < \ell_0$, \textbf{return} $\mathcal{S}$.
    \item If $|\mathcal{S}| \geqslant \ell_0$ or any pruning rule indicates that $\mathcal{S}$ cannot be extended to a quadratization of size $< \ell_0$, \textbf{return} \texttt{nothing}.
    \item\label{step:generation} Generate a list $\mathcal{S}_1, \ldots, \mathcal{S}_r$ of possible extensions of $\mathcal{S}$.
    \item Set $\mathcal{Q}_0 := \text{\texttt{nothing}}$, $\ell := \ell_0$.
    \item For $i = 1, \ldots, r$
    \begin{enumerate}
        \item Run the algorithm recursively on $\mathcal{S}_i$ and $\ell$.
        \item If the recursive call returns a quadratization $\mathcal{Q}$ with $|\mathcal{Q}| < \ell$,\\ set $\mathcal{Q}_0 := \mathcal{Q}$ and $\ell := |\mathcal{Q}|$.
    \end{enumerate}
    \item \textbf{Return} $\mathcal{Q}_0$.
\end{enumerate}
\end{algorithm}

\begin{remark}[On Laurent polynomials in \Qbee]\label{rem:laurent}
  We have extended the original \Qbee algorithm outlined above to allow the input system to be defined not only by polynomials but by \emph{Laurent} polynomials, that is, polynomials with possible negative degrees.
  This is done by allowing negative degrees at~\ref{step:generation}.
  In this case, we do not give termination and optimality guarantees but in practical examples (see~\cref{sec:exp_term} and~\cref{sec:heliospheric}) this approach works well.
  Therefore, all extensions of \Qbee presented here such as~\Cref{alg:qbee_input} and~\Cref{alg:qbee_da} can take Laurent polynomials as input but do not provide termination and optimality guarantees for this case.
  
  We could instead first polynomialize the model using our polynomialization algorithm from~\cref{sec:polynomialization} but this may result in a quadratization of larger dimension as the following example shows. 
  Consider a Laurent polynomial model
  \begin{equation}\label{eq:laurent_start}
    \dot{x}_1 = x_2^2, \quad \dot{x}_2 = x_1 x_2^{-1}.
  \end{equation}
  If we first polynomialize the model by introducing $x_3 = x_2^{-1}$, we obtain
  \begin{equation}\label{eq:laurent_poly}
    \dot{x}_1 = x_2^2, \quad \dot{x}_2 = x_1x_3, \quad \dot{x}_3 = -x_1x_3^3.
  \end{equation}
  Computation with \Qbee shows that at least two new variables are required to quadratize~\eqref{eq:laurent_poly}.
  On the other hand, \Qbee applied directly to~\eqref{eq:laurent_start} quadratizes the model with only two new variables $w_1 = x_2^{-1}$ and $w_2 = x_1 x_2^{-2}$:
  \[
    \dot{x}_1 = x_2^2, \quad \dot{x}_2 = x_1 w_1, \quad \dot{w}_1 = -w_1w_2, \quad \dot{w}_2 = 1 - 2 w_2^2.
  \]
  This discrepancy arises because, when quadratizing~\eqref{eq:laurent_poly}, \Qbee cannot take into account the fact that $x_2 x_3 = 1$ as it is not a part of its input.
\end{remark}

%%%%%%%%%%%%%%%%%%%%%%%%%%%%%%%%%%%%%%%%%%%%%%%%%%
\subsection{\Qbee for quadratization of polynomial ODEs with inputs}
Herein, we describe how the results of~\cref{sec:theory_inputs} can be used to add capabilities to~\cref{alg:qbee_outline} for computing quadratizations for systems with inputs. This requires modifying the initial upper bound $\ell_0$ and the procedure for computing the sets of variables for the recursive iterations at~\ref{step:generation}.
Consider the ODE system
\begin{equation}\label{eq:input_algo_start}
  \dot{\bx} = \bp(\bx, \bu),
\end{equation}
where $\bx = \bx(t) \in \real^\nn$ are the states and $\bu = \bu(t) \in \real^r$ are inputs and $\bp(\bx,\bu) = [p_1(\bx,\bu), \ldots, p_\nn(\bx,\bu)]^\top$ is a vector of polynomials.

\Cref{thm:existence_inputs} guarantees that it is always possible to find a quadratization with the new variables involving the inputs in at most zero order and, thus, the quadratized equations involving the inputs in at most first order (as in~\Cref{def:quadr_input}).
Since only zeroth and first order derivatives of inputs are involved in the quadratization process, we can impose an additional restriction that the inputs are linear\footnote{Formally, this can be justified as follows. We introduce the new `state' variables: $\bu^{(0)}=\bu$ and $\bu^{(1)} = \dot{\bu}$, as the zeroth and first derivatives of $\bu$, respectively. This allows us to formally write an \textit{autonomous} ODE system $\dot{\bx} = \bp(\bx, \bu^{(0)}), \dot{\bu}^{(0)} = \bu^{(1)}, \dot{\bu}^{(1)} = \mathbf{0}$, where the last equation acts as a dummy equation.  Since by \cref{def:quadr_input} the quadratizations of the non-autonomous system~\eqref{eq:input_algo_start} involve only $\bx$ and $\bu$ but not derivatives of $\bu$, there is a bijection between the quadratizations of~\eqref{eq:input_algo_start} and  the quadratizations of the autonomous system above involving $\bx$ and $\bu^{(0)}$ but not $\bu^{(1)}$.}, that is, $\ddot{\bu}(t) = 0$.
With this additional restriction, the system can be written as an autonomous ODE system to which the original~\Qbee algorithm can now be applied with two amendments: the bound $\ell_0$ on the size of quadratization is now taken from~\Cref{thm:existence_inputs}, not from~\Cref{thm:existence_simple}, and~\ref{step:generation} of~\Cref{alg:qbee_outline} is restricted to monomials in $\bx$ and $\bu$ (that is, using $\dot{\bu}$ is not allowed).
\Cref{thm:existence_inputs} guarantees that even under this restriction a quadratization exists, and \Qbee will find an optimal one.

To compute an input-free quadratization (for cases when $\bu(t)$ is not differentiable), the approach above is modified as follows.
First, we set $\ell_0 = \infty$.
Second, we further restrict~\ref{step:generation} of~\cref{alg:qbee_outline} by allowing only the extensions not involving either $\bu$ or $\dot{\bu}$.
We note that, in general, the search for input-free quadratizations is not guaranteed to terminate, as there may not be input-free quadratizations. 
We refer to~\cref{sec:input-free} for the criterion (\cref{prop:inp_general}) and special cases (\cref{lem:input_bilinear,lem:triangular}) where we expect the algorithm to terminate.
We summarize the algorithm explained above in~\cref{alg:qbee_input}.

\begin{algorithm}[H]
\caption{Extension of \Qbee to non-autonomous systems}
\label{alg:qbee_input}
\begin{description}[itemsep=0pt]
\item[Input] Non-autonomous polynomial ODE system $\dot{\bx} = \bp(\bx, \bu)$ with $\nn$ states and $r$ inputs and a Boolean flag $\texttt{input-free}$.
\item[Output] Optimal quadratization as defined in~\cref{def:quadr_input} if $\texttt{input-free} = \texttt{False}$ and optimal input-free quadratization as in~\cref{def:quadr_input_zero} if $\texttt{input-free} = \texttt{True}$ and such quadratization exists.
\end{description}

\begin{enumerate}[label = \textbf{(Step~\arabic*)}, leftmargin=*, align=left, labelsep=2pt, itemsep=4pt]
    \item Build an autonomous ODE system by restricting inputs to the linear ones, that is, adding $\ddot{u} = 0$ to the original system.
    \item\label{step:input_main} Run~\cref{alg:qbee_input} on the autonomous ODE system constructed in the previous step 
    \begin{itemize}
        \item If $\texttt{input-free} = \texttt{False}$, use $\ell_0$ computed as in \cref{thm:existence_inputs} and ensure that~\ref{step:generation} only selects the extensions not involving $\dot{\bu}$.
        \item If $\texttt{input-free} = \texttt{True}$, use $\ell_0 = \infty$ and and ensure that~\ref{step:generation} only selects the extensions not involving $\bu, \dot{\bu}$.
    \end{itemize}
    \item \textbf{Return} the quadratization found at the previous step.
\end{enumerate}
\end{algorithm}

\begin{example}[Use of \Qbee for systems with inputs]
Consider again the polynomial model with input from \cref{ex:infinite}. The \Qbee software finds a quadratization for that example with the following code.
\begin{minted}[frame=lines]{python}
    x, u = functions("x, u")
    system = [ # pairs of the form (x, x')
        (x, x**2 * u)
    ]
    quadratize(system).print() # input-free=False by default
\end{minted}
The code  produces: 
\begin{minted}[frame=lines]{text}
    Introduced variables:
    w0 = u*x

    x' = x*w0
    w0' = x*u' + w0**2  
\end{minted}

\end{example}

\begin{example}[Use of \Qbee for input-free quadratizations]
Consider again the polynomial model from \cref{ex:findim}. The \Qbee software finds an input-free quadratization for that example with the following code.
\begin{minted}[frame=lines]{python}
    x1, x2, u = functions("x1, x2, u")
    system = [ # pairs of the form (x, x')
        (x1, x1 + x1 * u),
        (x2, x1**2 * u)
    ]
    quadratize(system, input_free=True).print()
\end{minted}
In the code above, \mintinline{python}{input_free=True} means that we are looking for an input-free quadratization.
The code produces:
\begin{minted}[frame=lines]{text}
    Introduced variables:
    w0 = x1**2

    x1' = x1*u + x1
    x2' = u*w0
    w0' = 2*u*w0 + 2*w0
\end{minted}
\end{example}

%%%%%%%%%%%%%%%%%%%%%%%%%%%%%%%%%%%%%%%%%%%%%%%%%%%%%%%
\subsection{\Qbee for finding dimension-agnostic quadratizations}\label{sec:alg_agnostic}
In~\cref{sec:agnostic} we presented the theory for quadratizing families of linearly coupled ODE systems. With this, we mean that while the ODE dimension could be arbitrary, the equations have the same structure---so in essence, we can find a quadratization for a lower-dimensional ODE and this generalizes to adding more (similarly structured) equations. 

We revisit \cref{ex:PDE_quadratization} where we found a quadratization of a family of ODE systems of the form~\eqref{eq:pde_discr}:
\[
  \dot{\bx}^{[\nx]} = \bx^{[\nx]} + (\bx^{[\nx]})^2 \odot (\bD \bx^{[\nx]}).
\]
To provide such a family as an input to our algorithm, we formally write this ODE system as a scalar equation
\begin{equation}\label{eq:da_ex_input}
  \dot{x} = x + x^2 x_{\bD},
\end{equation}
where $x_{\bD}$ is a formal variable that we use as a placeholder for the linear coupling $\bD \bx^{[\nx]}$.
We extend this notation to the general case of a family of linearly coupled ODE systems as in~\cref{def:family}. We represent the family  symbolically as
\begin{equation}\label{eq:alg_da_input}
  \dot{\bx} = \bp(\bx, \bx_{\bD}),
\end{equation}
where $\bx = [x_1, \ldots, x_\nd]^\top$, $\bx_{\bD} = [x_{\bD, 1}, \ldots, x_{\bD, \nd}]^\top$ are again formal variables which are placeholders for the coupling expressions $(\bD_j \bx^{[\nx]}_{*, j})_i$, and $\bp(\bx, \bx_\bD)$ is affine in $\bx_\bD$, i.e., 
\begin{equation}\label{eq:pxx_expanded}
  \bp(\bx, \bx_\bD) = \bp_0(\bx) + \sum\limits_{i = 1}^\nd \bp_i(\bx)x_{\bD, i}.
\end{equation}
The vectors $\bp_0(\bx), \ldots, \bp_\nd(\bx)$ define the family of linearly coupled models as in~\cref{def:family}.

Given a family of linearly coupled ODE systems written in the form~\eqref{eq:alg_da_input}, \Cref{prop:agnostic} provides an effective condition to check whether given polynomial vectors $\bw_1(\bx)$ and $\bw_2(\bx, \tilde{\bx})$ provide a dimension-agnostic quadratization (see \cref{def:da}) of the family: it is sufficient to take a particular member of the family, $\cF_\bP^{[4]}(\bp_0, \ldots, \bp_\nd, \bD_1^\ast, \ldots, \bD_\nd^\ast)$ and check whether $\cM(\bw_1, \bw_2; \bD_1^\ast, \ldots, \bD_\nd^\ast)$ is a quadratization of this specific ODE system.
Moreover, we can use \cref{prop:agnostic} to find a dimension-agnostic quadratization for the family by running~\cref{alg:qbee_outline} on the discretization but restricting the extensions at~\ref{step:generation} to the ones of the form~$\cM(\bw_1, \bw_2; \bD_1^\ast, \ldots, \bD_\nd^\ast)$.
Success and termination of this search can be ensured by using the bounds for $\ell$ and $L$ from~\cref{thm:existencEgnostic} to provide the starting bound $\ell_0$ in~\cref{alg:qbee_outline}: since there are three off-diagonal elements in the $D_i$'s and $\nx = 4$, we take $\ell_0 = 4\ell + 3L = (3\nd + 4)\binom{\nd + d}{d}$.  This approach is summarized in~\cref{alg:qbee_da}.

\begin{example}[Use of \Qbee for dimension agnostic quadratization]\label{ex:use_agnostic}
Consider the family from~\Cref{ex:PDE_quadratization}.
We have already written the corresponding family of ODE systems in the form \eqref{eq:da_ex_input} required  by the algorithm.
\Qbee can be used to find a dimension-agnostic quadratization by writing the following code
\begin{minted}[frame=lines]{python}
    x, Dx = functions("x Dx")
    system = [
        (x, x + x**2 * Dx) 
    ]
    quadratize_dimension_agnostic(system)
\end{minted}
Then, \Qbee produces 
\begin{minted}[frame=lines]{text}
  Every ODE system in the family can be quadratized by adding the following variables
  * For each i, we take variables
  x_i**2
  * For every i and j such that the variables from the j-th block (that is, x_j) 
  appear in the right-hand sides of the i-th block (that is, in x_i'), we take
  x_i*x_j
\end{minted}
The quadratized system like~\eqref{eq:pde_ex_quadratized} in general depends on the shape of the $\bD$ matrices used (see~\Cref{ex:da_quadr_basic}), and it is not produced automatically but has to be derived by hand using the knowledge of the quadratizing variables. 
The function \mintinline{python}{quadratize_dimension_agnostic} has a keyword parameter \mintinline{python}{print_intermediate}, if the value is set to be \mintinline{python}{True}, the quadratization for the ODE system from~\Cref{prop:agnostic} is printed, which can be helpful when producing the quadratized system.
\end{example}

\begin{algorithm}[H]
\caption{Extension of \Qbee to dimension-agnostic quadratizations}
\label{alg:qbee_da}
\begin{description}[itemsep=0pt]
\item[Input] a family of linearly coupled ODE systems~\eqref{eq:from_template} presented as
\[
  \dot{\bx} = \bp(\bx, \bx_\bD),
\]
where $\bx = [x_1, \ldots, x_\nd]^\top$, $\bx_\bD = [x_{\bD, 1}, \ldots, x_{\bD, \nd}]$ are formal variables which are placeholder for the coupling, and $\bp(\bx, \bx_{\bD})$ is affine in $\bx_\bD$.
\item[Output] a dimension-agnostic quadratization of the family (see~\cref{def:da}).
\end{description}

\begin{enumerate}[label = \textbf{(Step~\arabic*)}, leftmargin=*, align=left, labelsep=2pt, itemsep=4pt]
    \item Set $d := \deg_{\bx} \bp(\bx, \bx_{\bD})$.
    \item Write $\bp(\bx, \bx_{\bD})$ in the form~\eqref{eq:pxx_expanded} and extract the vectors $\bp_0, \bp_1, \ldots, \bp_\nd$.
    \item Construct an ODE system $\cF_\bP^{[4]}(\bp_0, \ldots, \bp_\nd, \bD_1^\ast, \ldots, \bD_\nd^\ast)$ from~\cref{prop:agnostic}.
    \item\label{step:da_apply_qbee} Apply~\cref{alg:qbee_outline} to the produced ODE system with $\ell_0 = (3\nd + 4)\binom{\nd + d}{d}$ and~\ref{step:generation} selecting extensions of the form~$\cM(\bw_1, \bw_2; \bD_1^\ast, \ldots, \bD_\nd^\ast)$.
    \item From the produced quadratization, derive the polynomial vectors $\bw_1(\bx)$ and $\bw_2(\bx, \tilde{\bx})$ and \textbf{return} them.
\end{enumerate}
\end{algorithm}

Since \ref{step:da_apply_qbee} of \cref{alg:qbee_da} applies \Qbee to   $\cF^{[4]}(\bp_0, \ldots, \bp_\nd,\bD^*_1, \ldots, \bD_\nd^*)$, the resulting dimension-agnostic quadratization is not guaranteed to be optimal for all choices of $\nx$ and $\bD_1, \ldots, \bD_\nd$.  
In all the examples we have considered so far, and the forthcoming examples in~\cref{sec:applications,sec:heliospheric}, the dimension-agnostic quadratization produced by the algorithm was always of the same or lower dimension as the ones found by hand.

%%%%%%%%%%%%%%%%%%%%%%%%%%%%%%%%%%%%%%%%%%%%%%%%%%%%%%%
\subsection{Polynomialization as a pre-processing step for quadratization}\label{sec:polynomialization}
This work focuses on the task of quadratizing polynomial ODEs. 
However, many models in engineering and science are non-polynomial. The problem of optimal quadratization of non-polynomial systems remains open, but practical polynomialization algorithms can nevertheless be devised. 
We design and implement a polynomialization algorithm that can be used as a preprocessing step before quadratization with \cref{alg:qbee_outline} and~\cref{alg:qbee_input}.
The algorithm follows the general approach described in~\cite{gu2011qlmor,lifeware2}: at each step the algorithm finds a nonpolynomial term and adds a new variable corresponding to it. 
The main difference of our approach from the prior algorithms is that we again employ the branch-and-bound approach as in~\Cref{alg:qbee_outline}, that is, the algorithm does not perform a fixed sequence of substitutions but explores different sequences in a recursive fashion. 
While this may require more iterations, the resulting system can be of smaller dimensions, as the next example shows. 

\begin{example}\label{ex:poly_small}
    Consider the scalar nonpolynomial equation $\dot{x} = e^{-x} + e^{-2x}$. Here, the {\sc Biocham}~\cite{Biocham} software introduces two new variables $e^{-x}$ and $e^{-2x}$ by using the algorithm from~\cite{lifeware2}. 
    In contrast, our \Qbee algorithm finds that only one new variable, $w := e^{-x}$, is sufficient to make the system polynomial:
    \begin{align*}
        \dot{x} &= w + w^2,\\
        \dot{w} &= -w^2 - w^3.
    \end{align*}
    The following \Qbee code produces the above polynomialization of order one:

    \begin{minted}[frame=lines]{python}
    x = functions("x")
    system = [(x, sp.exp(-x) + sp.exp(-2 * x))] # list of pairs (x, x')
    print(polynomialize(system))
    \end{minted}
The \Qbee output is:
\begin{minted}[frame=lines]{text}
    w_0 = exp(-x)

    x' = w_0**2 + w_0
    w_0' = -w_0*(w_0**2 + w_0)
    \end{minted}
\end{example}

\subsection{On optimality guarantees of the presented algorithms}\label{sec:optimality}
We summarize the optimality guarantees which are available for the algorithms we presented herein.
\begin{itemize}
    \item If the input is a polynomial ODE system of the form~\eqref{eq:sys_main} or~\eqref{eq:sys_main_input}, then the original \Qbee (\Cref{alg:qbee_outline}) and its extension to the non-autonomous systems (\Cref{alg:qbee_input}) are guaranteed to produce an optimal monomial quadratization.
    \item For the polynomialization algorithm (\Cref{sec:polynomialization}), the algorithm for computing dimension-agnostic quadratizations (\Cref{alg:qbee_da}), as well as \Cref{alg:qbee_outline,alg:qbee_input} when allowing for Laurent polynomials in the model form, we do not guarantee optimality.
    However, our algorithms always reduce the number of new variables by applying the branch-and-bound approach from \Qbee and, in particular, manage to produce quadratizations of smaller dimensions than it was possible before (see~\cref{sec:applications,sec:heliospheric})
\end{itemize}

%%%%%%%%%%%%%%%%%%%%%%%%%%%%%%%%%%%%%%%%%%%%%%%%%%%%%%%
\section{Applications of quadratization} \label{sec:applications}
%%%%%%%%%%%%%%%%%%%%%%%%%%%%%%%%%%%%%%%%%%%%%%%%%%%%%%%
This section applies the \Qbee algorithm to quadratize a wide range of nonlinear systems of ODEs.  
Here, we revisit examples from the literature and demonstrate that we can improve the order of quadratization (i.e., we need fewer variables to do so). \Cref{sec:tubularReactor} considers two different tubular reactor models, both with variable dimension and external inputs
In \cref{sec:gems} we consider a two-step rocket combustion process which is nonpolynomial.

%%%%%%%%%%%%%%%%%%%%%%%%%%%%%%%%%%%%%%%%%%%%%%%%%
\subsection{Tubular reactor model} \label{sec:tubularReactor}
%%%%%%%%%%%%%%%%%%%%%%%%%%%%%%%%%%%%%%%%%%%%%%%%%
We consider a non-adiabatic tubular reactor model with a single reaction as in~\cite{heinemann1981tubular} and follow the discretization and setup in \cite{zhou2012thesis,KW2019_balanced_truncation_lifted_QB}. 
The spatially-discretized model describes the evolution of the species concentration vector $\bm{\psi}(t)$ and temperature vector $\bm{\theta}(t)$ over time $t >0$.  
We consider two cases of different complexity.
Both cases are non-autonomous and of variable dimension (depending on the size of discretization).

%%%%%%%%%%%%%%%%%%%%%%%%%%%%%%%%%%%%%%%%%%%%%%%%
\subsubsection{Polynomial reaction model}

We consider the polynomial ODE in $2\nx$ (i.e., $\nd = 2$) dimensions:
\begin{equation}
\label{eq:main_poly_reaction}
\begin{split}
  \bm{\dot{\psi}}  &= \mathbf{A}_\psi \bm{\psi} + \mathbf{b}_\psi - \mathcal{D}\bm{\psi} \odot (\mathbf{c}_0 + \mathbf{c}_1 \odot\bm{\theta} + \mathbf{c}_2 \odot\bm{\theta}^2 + \mathbf{c}_3 \odot\bm{\theta}^3), \\
  \bm{\dot{\theta}}  &= \mathbf{A}_\theta \bm{\theta} + \mathbf{b}_\theta + \mathbf{b} u + \mathcal{B} \mathcal{D} \bm{\psi} \odot (\mathbf{c}_0 + \mathbf{c}_1 \odot\bm{\theta} + \mathbf{c}_2 \odot\bm{\theta}^2 + \mathbf{c}_3 \odot\bm{\theta}^3),
\end{split}
\end{equation}
where $u=u(t)$ is a scalar input function, $\mathbf{A}_{\psi}$ and $\mathbf{A}_{\theta}$ are $\nx\times \nx$ constant matrices, $\mathbf{b}, \mathbf{b}_\psi,$ and $\mathbf{b}_\theta$ are $\nx$-dimensional constant vectors. 
The Damk{\"o}hler number $\mathcal{D}$ and $\mathcal{B}$ are the known constants. 
This ODE system has a polynomial reaction model, i.e., the chemical reaction term is $f(\bm{\psi}, \bm{\theta}; \gamma) = \psi (c_0 + c_1\bm{\theta} + c_2\bm{\theta}^2 + c_3\bm{\theta}^3)$ following~\cite[Eqs. (17)-(18)]{KW2019_balanced_truncation_lifted_QB}.  
In~\cite{KW2019_balanced_truncation_lifted_QB} this model was quadratized with $5\nx$ additional variables, so the lifted system dimension was $7\nx$.

In order to use \Qbee, we rewrite~\eqref{eq:main_poly_reaction} in the form~\eqref{eq:alg_da_input}, that is, as a formal three-dimensional ODE system with additional variables $\psi_{\bD}$ and $\theta_\bD$ instead of $\bA_{\psi}\bm{\psi}$ and $\bA_{\theta}\bm{\theta}$:
\begin{equation}
\begin{split}
  \dot{\psi}  &= \psi_{\bD} + b_\psi - \mathcal{D}\psi (c_0 + c_1 \theta + c_2 \theta^2 + c_3 \theta^3), \\
  \dot{\theta}  &= \theta_\bD + b_\theta + b u + \mathcal{B} \mathcal{D} \psi (c_0 + c_1 \theta + c_2 \theta^2 + c_3 \theta^3),
\end{split}
\end{equation}
Then we use \Qbee to search for a discretization-agnostic quadratization as described in~\cref{sec:alg_agnostic}. 
\Qbee finds that for every $\nx$ and for every matrices $\bA_{\psi}$ and $\bA_{\theta}$, the system can be quadratized using the following $4\nx$ additional variables:
\begin{equation}\label{eq:quadr_poly_reaction}
    \bm{\bw}_1 := \bm{\theta}^2, \quad \bm{\bw}_2 := \bm{\theta}^3, \quad \bm{\bw}_3 := \bm{\psi}\odot \bm{\theta}, \quad \bm{\bw}_4 := \bm{\psi}\odot \bm{\theta}^2.
\end{equation}
One can use these new variables to write the quadratized system:

\begin{equation}
\begin{split}
    \dot{\bm{\psi}} & = \bA_{\psi}\bm{\psi} + \bm{b}_{\psi} -\mathcal{D} (\bc_0\odot \bm{\psi} + \bc_1\odot \bw_3 + \bc_2\odot \bw_4 + \bc_3\odot \bw_1\odot \bw_3),\\
    \dot{\bm{\theta}} & = \bA_{\theta}\bm{\theta} + \bm{b}_{\theta} + \bm{b} u + \mathcal{BD}(\bc_0 \odot \bm{\psi} + \bc_1 \odot \bw_3 + \bc_2 \odot \bw_4 + \bc_3 \odot \bw_1 \odot \bw_3),\\
    \dot{\bw}_1 & = 2\bm{\theta} (\bA_{\theta}\bm{\theta} + \bm{b}_{\theta} + \bm{b} u) + 2\mathcal{BD} (\bc_1 \odot \bw_4 + \bw_3\odot(\bc_0 + \bc_2 \odot \bw_1 + \bc_3 \odot \bw_2)),\\
    \dot{\bw}_2 & = 3\bm{\theta}(\bA_{\theta}\bm{\theta} + \bm{b}_{\theta} + \bm{b} u) + 3\mathcal{BD}(\bw_4\odot(\bc_0 + \bc_3 \odot \bw_2) + \bw_3\odot(\bc_1 \odot \bw_1 + \bc_2 \odot \bw_2)),\\
    \dot{\bw}_3 & = \bm{\psi}(\bA_{\theta}\bm{\theta} + \bm{b}_{\theta} + \bm{b} u) + \bm{\theta}(\bA_{\psi}\bm{\psi} + \bm{b}_{\psi}) - \mathcal{D}(\bc_1 \odot \bw_4 + \bw_3\odot(\bc_0 + \bc_2 \odot \bw_1 + \bc_3 \odot \bw_2))\\ & \quad + \mathcal{BD}(\bc_0 \odot \bm{\psi}^2 + \bw_3\odot(\bc_1 \odot \bm{\psi} + \bc_2 \odot \bw_3 + \bc_3 \odot \bw_4)),\\
    \dot{\bw}_4 & = \bw_1(\bA_{\psi}\bm{\psi} + \bm{b}_{\psi}) + 2\bw_3(\bA_{\theta}\bm{\theta} + \bm{b}_{\theta} + \bm{b}u) - \mathcal{D}(\bc_0 \odot\bw_4 + \bc_1\odot \bw_1 \odot \bw_3 + \bw_2\odot(\bc_2 \odot \bw_3 + \bc_3 \odot \bw_4))\\
    &  \quad + 2\mathcal{BD}(\bc_0 \odot \bm{\psi} \odot \bw_3 + \bc_1 \odot \bw_3^2 + \bc_2 \odot \bw_3 \odot \bw_4 + \bc_3 \odot \bw_4^2).
\end{split}
\end{equation}
Thus, \Qbee produced a lower-dimensional quadratization ($4\nx$ extra variables) than has been previously reported in the literature ($5\nx$ extra variables).

%%%%%%%%%%%%%%%%%%%%%%%%%%%%%%%%%%%%%%%%%%%%%%%%%%%%%%%%%%%%%%%%%%%%%%%%%%%%%
\subsubsection{Exponential reaction terms}\label{sec:exp_term}
We now consider a non-polynomial non-autonomous ODE system. The model is taken from~\cite[Sec. V]{KW18nonlinearMORliftingPOD}, and includes an Arrhenius reaction term that is of exponential kind, i.e., the system of non-polynomial non-autonomous ODEs is
\begin{equation}\label{eq:full}
    \begin{split}
        \bm{\dot{\psi}} &= \mathbf{A}_{\psi} \bm{\psi} + \mathbf{b}_\psi - \mathcal{D} \bm{\psi} \odot e^{\gamma - \gamma / \bm{\theta}},\\
        \bm{\dot{\theta}} &= \mathbf{A}_{\theta} \bm{\theta} + \mathbf{b}_\theta + \mathbf{b} u(t) + \mathcal{B} \mathcal{D} \bm{\psi} \odot e^{\gamma - \gamma / \bm{\theta}}.
    \end{split}
\end{equation}
In \cite{KW18nonlinearMORliftingPOD} two lifting transformations are derived: (i) the system is polynomialized with $3\nx$ extra variables, $\bw_1 := \btheta^{-1}, \bw_2 :=\btheta^{-2}$ and $\bw_3 := e^{\gamma - \gamma / \btheta}$; (ii) the system is further transformed into a quadratic system of differential-algebraic equations (DAEs) by adding $3\nx$ additional variables, thus, adding $6\nx$ variables in total.

With \Qbee, we improve on these results. First, we manually introduce only one variable $\bw_1 := e^{\gamma - \gamma / \btheta}$ and obtain a $3\nx$-dimensional system with Laurent polynomials in the right-hand side
\begin{equation}\label{eq:polynomialized_exp_reaction}
  \begin{split}
      \bm{\dot{\psi}} & = \mathbf{A}_{\psi} \bm{\psi} + \mathbf{b}_\psi - \mathcal{D} \bm{\psi} \odot \mathbf{w}_1,\\
      \bm{\dot{\theta}} & = \mathbf{A}_{\theta} \bm{\theta} + \mathbf{b}_\theta + \bb u(t) + \mathcal{B} \mathcal{D} \bm{\psi} \odot \mathbf{w}_1,\\
      \dot{\mathbf{w}}_1 & = \gamma\bm{\dot{\theta}} \odot \frac{1}{\bm{\theta}^2} \odot \bm{w}_1.
  \end{split}
\end{equation}
In order to use \Qbee, we rewrite~\eqref{eq:polynomialized_exp_reaction} in the form~\eqref{eq:alg_da_input}, that is, as a formal three-dimensional ODE system with additional variables $\psi_{\bD}$ and $\theta_\bD$ instead of $\bA_{\psi}\bm{\psi}$ and $\bA_{\theta}\bm{\theta}$:
\begin{align*}
  \dot{\psi} &= \psi_{\bD} + b_{\psi} - \mathcal{D} \psi w_1,\\
  \dot{\theta} &= \theta_{\bD} + b_{\theta} u + \mathcal{B}\mathcal{D} \psi w_1,\\
  \dot{w}_1 &= \gamma \frac{w_1}{\theta^2}(\theta_{\bD} + b_{\theta} + b u + \mathcal{B}\mathcal{D} \psi w_1).
\end{align*}
Then we use \Qbee to search for a discretization-agnostic quadratization as described in~\cref{sec:alg_agnostic}. 
As explained in~\Cref{rem:laurent}, this can be done even though~\eqref{eq:polynomialized_exp_reaction} is a Laurent polynomial system.
\Qbee finds that for every $n$ and for every matrices $\bA_{\psi}$ and $\bA_{\theta}$, a quadratization requires the additional variables
    \begin{align*}
        \bw_2 &:= \frac{1}{\bm{\theta}},\; &\bw_4 &:= \frac{u(t)}{\bm{\theta}},\; &\bw_8 &:= \frac{\bm{\psi} \odot \bw_1}{\bm{\theta}^2},\\
        \bw_3 &:= \frac{1}{\bm{\theta}^2},\;  &\bw_5 &:= \frac{u(t)}{\bm{\theta}^2},\; &\bw_7 &:= \frac{\bm{\psi}\odot \bw_1}{\bm{\theta}}.
    \end{align*}
and we additionally require for every $1 \leqslant i \neq j \leqslant \nx$ where $(\bA_\psi)_{i, j} \neq 0$ or $(\bA_\theta)_{i, j} \neq 0$:
    \[
      w_{8, i, j} := \frac{\psi_j}{\theta_i},\quad w_{9, i, j} := \frac{\psi_j}{\theta_i^2} ,\quad w_{10, i, j} := \frac{\theta_j}{\theta_i},\quad w_{11, i, j} := \frac{\theta_j}{\theta_i^2}.
    \]
Therefore, if we denote by $M$ the number of nonzero off-diagonal elements in $\bm{A}_{\theta}$ and $\bA_{\psi}$, we add $7\nx + 4M$ new variables.
While these are more additional variables than the aforementioned $6\nx$ new variables from~\cite{KW18nonlinearMORliftingPOD}, our transformation yields an ODE (not a DAE) which is generally much more advantageous to work with for analysis, simulation, and control. 

%%%%%%%%%%%%%%%%%%%%%%%%%%%%%%%%%%%%%%%%%%%%%%%%%%%%%%%%%%%%%
\subsection{Species' reaction model for combustion}\label{sec:gems}
%%%%%%%%%%%%%%%%%%%%%%%%%%%%%%%%%%%%%%%%%%%%%%%%%%%%%%%%%%%%%%%%%%%%%%%%%
We consider the rocket engine combustion model from~\cite[Eqs. (A1a-A1d)]{SKHW2020_learning_ROMs_combustor}.
At each spatial grid point, the species' molar concentrations $x_1, x_2, x_3$ and $x_4$ are determined by the following  equations:
\begin{equation}\label{eq:main_gems}
\begin{split}
  \dot x_1 &= -A \exp{\Big(-\frac{E}{R u(t)}\Big)} x_1^{0.2} x_2^{1.3}\\
  \dot x_2 &= 2 \dot x_1\\
  \dot x_3 &= -\dot x_1 \\
  \dot x_4 &= -2 \dot x_1,
\end{split}
\end{equation}
where $u(t)$ is a time dependent input standing for the temperature, and $A,\ E$ and $R$ are known parameters.
Our goal is to compute a polynomialization and a quadratization via \Qbee and compare the results to  \cite{SKHW2020_learning_ROMs_combustor} where this was done by hand. 
$\Qbee$ performs polynomialization and quadratization of the model automatically. 
It first finds three variables which allow to lift the write the system using Laurent polynomials
\[
 w_1 = x_2^{1.3}, \quad w_2 = x_1^{0.2}, \quad w_3 = e^{-E / (Ru(t))}
\]
The lifted system will be
\begin{align*}
    \dot{x}_1 &= -A w_1 w_2 w_3,
    &\dot{x}_2 &= -2 A w_1 w_2 w_3,
    &\dot{x}_3 &= A w_1 w_2 w_3,
    &\dot{x}_4 &= 2 A w_1 w_2 w_3,\\
    \dot{w}_1 &= -2.6 A w_1^2 w_2 w_3 x_2^{-1},
    &\dot{w}_2 &= -0.2 A w_1 w_2^2 w_3 x_1^{-1},
    &\dot{w}_3 &= \frac{E \dot{u}(t) w_3}{R u^2(t)} & &
\end{align*}
Then this system is quadratized using seven more variables
\begin{align*}
  w_4 &= w_1 w_2,\quad &w_5 &= \frac{\dot{u}(t)}{u(t)^2},\quad &w_6 &= \frac{1}{u(t)^2}, \quad &w_7 &= \frac{\dot{u}(t)}{u(t)},\\
  \quad w_8 &= \frac{1}{u(t)}, \quad &w_9 &= w_1 w_2 w_3 x_1^{-1}, \quad &w_{10} &= w_1 w_2 w_3 x_2^{-1}. & &
\end{align*}
The resulting quadratic system is:
\begin{align*}
\dot{x}_1 &= -A w_1 w_4,
&\dot{w}_2 &= -0.2 A w_2 w_9,
&\dot{w}_7 &= \ddot{u}(t) w_8 - w_7^2,\\
\dot{x}_2 &= -2 A w_1 w_4,
&\dot{w}_3 &= \frac{E w_3 w_5}{R},
&\dot{w}_8 &= -w_7 w_8,\\
\dot{x}_3 &= A w_1 w_4,
&\dot{w}_4 &= -0.2 A w_4 w_9 + \frac{E w_4 w_5}{R},
&\dot{w}_9 &= Aw_9(0.8 w_9 - 2.6 w_{10} + \frac{Ew_5}{AR}),\\
\dot{x}_4 &= 2 A w_1 w_4,
&\dot{w}_5 &= \ddot{u}(t) w_6 - 2 w_5 w_7,
&\dot{w}_{10} &= -Aw_{10}(0.2w_9 + 0.6 w_{10} - \frac{E w_5}{AR}),\\
\dot{w}_1 &= -2.6 A w_1 w_{10},
&\dot{w}_6 &= -2 w_6 w_7.
\end{align*}
In total, we need to add ten new variables to the original system~\eqref{eq:main_gems} to obtain a quadratic ODE system.  The authors in \cite{SKHW2020_learning_ROMs_combustor} did not find a quadratized form for this system by hand due to the complexity of finding such quadratizations. This illustrates a significant advantage of the automated polynomialization and quadratization implemented in \Qbee.
Furthermore, one can observe that the equations for the new variables do not involve $x_1$ and $x_2$, and $x_1$ and $x_2$ can be expressed as $w_1w_2w_3w_9^{-1}$ and $w_1w_2w_3w_{10}^{-1}$, respectively.
Thus, one can omit $x_1$ and $x_2$ and work with only 12-dimensional system.

This example exhibits the same subtlety as the one~\Cref{rem:laurent}: if we first transformed the system to a polynomial (not Laurent polynomial) one and then quadratized, then we would add eleven variables.
We also note that it is not possible to find a quadratization without the appearance of $\ddot{u}(t)$ in the resulting system,
this can be deduced by applying~\Cref{prop:inp_general} to the polynomialized system.

%%%%%%%%%%%%%%%%%%%%%%%%%%%%%%%%%%%%%%%%%%%%%%%%%%%%%%%%%%%%%%%%
\section{Model learning for solar wind prediction with quadratized variables} \label{sec:heliospheric}
%%%%%%%%%%%%%%%%%%%%%%%%%%%%%%%%%%%%%%%%%%%%%%%%%%%%%%%%%%%%%%%
In this section, we illustrate the advantages of quadratization in the use-case of data-driven reduced-order modeling for solar wind prediction. The benefit of a quadratization of the dynamical system is that it provides a direct parametrization of the model and avoids hyperreduction, as discussed in the introduction. Thus, the free parameters (the coefficients matrices for the polynomial terms) can be learned from trajectory data; see  \cite{SKHW2020_learning_ROMs_combustor,QKMW2019_transform_and_learn,QKPW2020_lift_and_learn,JQK2021_performanceCompCombustion,SKP_PIregulartizationOPINF} for a variety of applications of this approach. 
\Cref{sec:HUX_model} briefly introduces the model, 
\cref{sec:HUXimplementation} gives implementation details, \cref{sec:HUXquadratizaton} presents the results of the \Qbee quadratization,  and \cref{sec:HUX-numerics-quadratic} presents the numerical result on the simulated quadratized model.

%%%%%%%%%%%%%%%%%%%%%%%%%%%%%%%%%%%%%%%%%%%%%%%%%%%%%%%%%%%%%%%%%%%
\subsection{Model} \label{sec:HUX_model}
The Heliospheric Upwind Extrapolation (HUX) model~\cite{riley_HUXP1_2011} is a two-dimensional nonlinear scalar homogeneous time-stationary PDE of the solar wind radial velocity in the heliospheric domain, where the independent variables are the radial distance from the Sun $r$ and Carrington longitude $\phi$ and the dependent variable is the solar wind velocity in the radial direction $v(r, \phi)$. The angular frequency of the Sun's rotation is evaluated at a constant Carrington latitude $\hat{\theta}$; if we consider the Sun's equatorial plane ($\hat{\theta} = 0$), then $\Omega_{\text{rot}}(0) = \frac{2 \pi}{25.38} \text{1/days}$ at the solar equator. The initial condition $v(r_{0}, \phi) = v_{0} (\phi)$ is defined on the periodic domain $0 \leq \phi \leq 2 \pi$ and $r \geq 0.14 \text{AU}$.

After semi-discretization via the upwind scheme, and an $r$-dependent linear shift of the longitude to account for advection (see~\cite{IK22_solarWind_sOPINF} for details), the finite-dimensional nonlinear model is
\begin{equation}\label{shifted-dynamics-ode-hux}
    \frac{\text{d}\bv(r)}{\text{d}r} = \bD\ln \left[\bv(r)\right] - \frac{\xi}{\Omega_\text{rot}(\hat{\theta})}\bD \bv(r),
\end{equation}
where $\bv(r)= [v(r, \phi_{1}), v(r, \phi_{2}), \ldots, v(r, \phi_{\nx})]^{\top} \in \real^{\nx}$  is the state vector discretized over $\nn$ points in longitude at heliocentric distance $r$, $\xi$ is the shift velocity which is fixed for a given Carrington Rotation, and the sparse matrix $\bD \in \real^{\nx \times \nx}$ is 
\begin{equation} \label{D-matrix}
    \bD = \frac{\Omega_{\text{rot}}(\hat{\theta})}{\Delta \phi}
    \left[
    \begin{array}{cccccc}
        -1 & 1 & 0 & & &\\
        0 & -1 & 1 & & & \\
        & \ddots & \ddots & \ddots & \ddots &\\
        & & & -1& 1 & 0 \\
        & & & & -1 & 1 \\
        1 & & & & 0 & -1 \\
    \end{array}
    \right] \in\real^{\nx \times \nx}.
\end{equation}

%%%%%%%%%%%%%%%%%%%%%%%%%%%%%%%%%%%%%%%%%%%%%
\subsection{Quadratization of the model} \label{sec:HUXquadratizaton}
To obtain a system of quadratic ODEs, we start with the original model~\eqref{shifted-dynamics-ode-hux} and eliminate the logarithmic function by adding a new variable $\bw_0 := \ln (\bv)$.
Let $C_1 := \frac{\xi}{\Omega_{\text{rot}}(\hat{\theta})}$, so that in the variables $\bv$ and $\bw_0$, the system~\eqref{shifted-dynamics-ode-hux} becomes
\begin{equation}\label{eq:hux_poly}
\begin{split}
  \frac{\text{d} \bv}{\text{d}r} &= \bD \bw_0 - C_1 \bD \bv,\\
  \frac{\text{d} \bw_0}{\text{d}r} &= \frac{1}{\bv}\bD \bw_0 - \frac{C_1}{\bv}\bD \bv.
\end{split}
\end{equation}
In order to use \Qbee, we rewrite~\eqref{eq:hux_poly} in the form~\eqref{eq:alg_da_input}, that is, as a formal two-dimensional ($\nd = 2$) ODE system with additional variables $v_{\bD}$ and $w_{0, \bD}$ instead of $\bD\bv$ and $\bD \bw_0$:
\begin{align*}
    \dot{v} &= w_{0, \bD} - C_1 v_{\bD},\\
    \dot{w}_0 &= \frac{w_{0, \bD}}{v} - C_1 \frac{v_{\bD}}{v}.
\end{align*}
Then we use \Qbee to search for a discretization-agnostic quadratization of~\eqref{eq:hux_poly} as described in~\cref{sec:agnostic}.
As explained in~\Cref{rem:laurent}, this can be done even though~\eqref{eq:hux_poly} is a Laurent polynomial system.
The dimension-agnostic quadratization returned by~\Cref{alg:qbee_da} is
\[
  \bw_1 = \left [\frac{1}{v}, \frac{w_0}{v}\right ]^\top \quad \text{ and }  \quad \bw_2 = \left [ \frac{\tilde{v}}{v}, \frac{\tilde{w}_0}{v} \right ]^\top.
\]
We can now specialize this dimension-agnostic quadratization to the matrix $\bD$ from~\eqref{D-matrix}.
Since the off-diagonal entries of $\bD$ are on the shifted diagonal, $\tilde{v}$ and $\tilde{w}_0$ will be replaced with $\bS\bv$ and $\bS\bw_0$, respectively, where $\bS$ denotes the cyclic shift operator sending any vector $[a_1, \ldots, a_\nx]^\top$ to $[a_\nx, a_1, \ldots, a_{\nx - 1}]^\top$.
Thus, we will obtain the following $4\nx$ variables:
\[
  \bw_{1, 1} = \frac{1}{\bv}, \quad \bw_{1, 2} = \frac{\bw_0}{\bv}, \quad \bw_{2, 1} = \frac{\bS\bv}{\bv}, \quad \bw_{2, 2} = \frac{\bS\bw_0}{\bv}.
\]
Let $C_2 := \frac{\Omega_{\text{rot}}(\hat{\theta})}{\Delta \phi}$, then the  quadratic system in these new variables is
\begin{equation}\label{eq:sw_quadr_orig}
\begin{split}
  \frac{\text{d}\bw_{1, 1}}{\text{d}r} &= -C_2 \bw_{1, 1} \odot \bW,\\
  \frac{\text{d}\bw_{1, 2}}{\text{d}r} &= C_2\bW \odot (\bw_{1, 1} - \bw_{1, 2}),\\
  \frac{\text{d}\bw_{2, 1}}{\text{d}r} &= C_2\bw_{2, 1} \odot (\bS\bW - \bW),\\
  \frac{\text{d}\bw_{2, 2}}{\text{d}r} &= C_2 \bw_{1, 1} \odot \bS\bW - C_2 \bw_{2, 2} \odot \bW,
\end{split}
\end{equation}
where we denote $\bW = \bw_{2, 2} - \bw_{1, 2} + C_1 \mathbf{1} - C_1 \bw_{2, 1}$ and $\mathbf{1} = [1, \ldots, 1]^\top \in \mathbb{R}^\nx$.
Furthermore, we see that the equations~\eqref{eq:sw_quadr_orig} do not involve $\bv$ and $\bw_0$.
Since the values of both $\bv$ and $\bw_0 = \ln(\bv)$ can be computed from $\bw_{1, 1} = \frac{1}{\bv}$, we henceforth only work with the equations~\eqref{eq:sw_quadr_orig} in variables $\bw_{1, 1}, \bw_{1, 2}, \bw_{2, 1}, \bw_{2, 2}$.
We observe that system~\eqref{eq:sw_quadr_orig} can be further reduced by hand via introducing $\bw_{3} = \bw_{2, 2} - \bw_{1, 2}$. 
Then $\bW$ can be written as $\bw_{3} + C_1\mathbf{1} - C_1\bw_{2, 1}$ and
\[
    \frac{\text{d}\bw_{3}}{\text{d}r} = C_2 \bw_{1, 1} \odot \bS \bW - C_2 \bw_{2, 2} \odot \bW - C_2\bW \odot (\bw_{1, 1} - \bw_{1, 2}) = C_2\bw_{1, 1} \odot (\bS \bW - \bW) - C_2\bw_{3} \odot \bW.
\]
Therefore, we can replace the $4\nx$-dimensional system~\eqref{eq:sw_quadr_orig} with a $3\nx$-dimensional
\begin{equation} \label{eq:HUX_Final_QB}
\begin{split}
  \frac{\text{d}\bw_{1, 1}}{\text{d}r} &= -C_2 \bw_{1, 1} \odot \bW,\\
  \frac{\text{d}\bw_{2, 1}}{\text{d}r} &= C_2\bw_{2, 1} \odot (\bS\bW - \bW),\\
  \frac{\text{d}\bw_{3}}{\text{d}r} &= C_2\bw_{1, 1}( \bS \bW - \bW) - C_2\bw_{3} \odot \bW,
\end{split}
\end{equation}
where $\bW = \bw_{3} + C_1\mathbf{1} - C_1\bw_{2, 1}$.
Note, that this system still contains $\bw_{1, 1} = \frac{1}{\bv}$, so the original trajectory for $v$ can be reconstructed once solutions to \eqref{eq:HUX_Final_QB} is solved. We can rewrite this system compactly in quadratic form as
\begin{equation}\label{cubic-shift-lift-fom}
    \dot{\bx}= \bA \bx + \bH \left(\bx \otimes^{\prime} \bx\right),
\end{equation}
where $\otimes^{\prime}$ denotes the compact Kronecker product, $\bx = \bx(r) = \left[\begin{matrix} \bw_{1, 1}^\top(r) & \bw_{2, 1}^\top(r) & \bw_3^\top(r) \end{matrix}\right]^{\top}\in \mathbb{R}^{3\nx}$, $\bA \in \real^{3\nx \times 3\nx}$ is the linear operator, and $\bH \in \real^{3\nx \times \frac{1}{2}(3\nx)(3\nx + 1)}$ is the quadratic operator. 
In sum, we started with an $\nx$-dimensional nonpolynomial (with logarithmic terms) system \eqref{shifted-dynamics-ode-hux} and through \Qbee are able to replace it with a $3\nx$-dimensional quadratic system \eqref{eq:HUX_Final_QB}. 

%%%%%%%%%%%%%%%%%%%%%%%%%%%%%%%%%%%%%%%%%%%%%%
\subsection{Data-driven reduced-order model implementation details} \label{sec:HUXimplementation}
We derived the lifted system in \eqref{eq:HUX_Final_QB} with the goal to learn a reduced-order model (ROM) from simulated trajectory data--we do not simulate the lifted system. We learn a ROM via the Operator Inference~\cite{peherstorfer2016data} framework and use the Python package version 1.2.1~\cite{OpInfGitHub} to implement the model learning and prediction. The numerical results in this section are shown for Carrington Rotation 2210, which occurred from 26 October to 23 November 2018, during solar minimum. For more details of the ROM method for this solar wind application, see~\cite{IK22_solarWind_sOPINF}. 

We next provide some numerical details. The state vector of the full-order model which we simulate to generate data is $\bv\in \mathbb{R}^{129}$. We take $n_{r} = 400$ uniform steps in the $r$ variables for discretization. We simulate the full-order model via the forward-Euler scheme and the ROM via an implicit multi-step variable method based on a backward differentiation formula using the \texttt{scipy.integrate.solve\_ivp()} Python function.
The implemented method is based on a shift of the independent variable to account for advection. We compute the shift function $c(r)$ and its derivative $\xi$ via the cross-correlation extrapolation method, see~\cite{IK22_solarWind_sOPINF} for more details. Note that the training data is obtained by shifting $\bv(r)$ instead of directly solving the shifted-lifted equations. 
Note that the non-lifted ROM is trained on velocity data in units of [km/s] and the lifted ROM is trained on data in units of [AU/days]. We choose to change the units in order to avoid very large or small scales of the lifted variables. The change in units also indicates the change in the scales of the regularization coefficients.

%%%%%%%%%%%%%%%%%%%%%%%%%%%%%%%%%%%%%%%%%%%
\subsection{Numerical results} \label{sec:HUX-numerics-quadratic}
Here, we seek to give data to answer the following question: \textit{``Does lifting the nonlinear system Eq.~\eqref{shifted-dynamics-ode-hux} to quadratic form, i.e. Eq.~\eqref{eq:HUX_Final_QB}, improve the ROM learning?''}
We set to answer this question by comparing the numerical results of a quadratic ROM learned from trajectories of the original states, $\bv$ from Eq.~\eqref{shifted-dynamics-ode-hux}, to a quadratic ROM learned from the lifted state variables, $\bw_{1, 1}, \bw_{2, 1}, \bw_3$ in Eq.~\eqref{eq:HUX_Final_QB}.

We use 70\% of our total snapshots (280 snapshots) for training to compute the basis of the low-dimensional model and 30\% for testing (120 snapshots). Therefore, the training domain is $[0.14 \text{AU}, 0.813 \text{AU}]$ and the testing domain is $[0.813 \text{AU}, 1.1 \text{AU}]$.
We choose the reduced dimension to be $\ell = 6$ for all learned ROMs. 
The ROM learning framework requires regularization for each of the learned matrices/tensors, i.e., $\bA$ and $\bH$ require regularization coefficients $\lambda_{1},\lambda_{2}$. We choose those from the logarithmically spaced set, i.e. $\{10^0, 10^1, \ldots, 10^9\}$, such that the best coefficients minimize the mean relative error over the training regime. The optimal coefficients for the shifted ROM in the training regime are $\{\lambda_{1} = 1, \lambda_{2} = 10^4\}$ and the optimal coefficients for the shifted lifted ROM in the training regime are $\{\lambda_{1} = 1, \lambda_{2} = 10\}$.

\Cref{fig:quadratic-ROM-solutions} shows the solutions obtained from both quadratic ROMs, the one learned from non-lifted and the one from lifted data. In both cases, the numerical results show that overall both methods are accurate in the training data (less than 1\% relative error). However, the ROM learned from the original data becomes unstable and produces spurious solutions, (note the different error bars). 
We compare the relative error in the training and testing regime in \cref{fig:quadratic-ROM-errors}. While both methods are almost equally accurate in the training domain, the quadratic ROM learned from the original variables produces unstable results in the testing data---is apparent that lifting the dynamics improves the ROM's accuracy in the testing regime. In conclusion, for predictive computation, the variable lifting makes a significant difference.
% %%%%%%%%%%%%%%%%%%%%%%%%%%%%%%%%
\begin{figure}
\centering
    \includegraphics[width=1\linewidth]{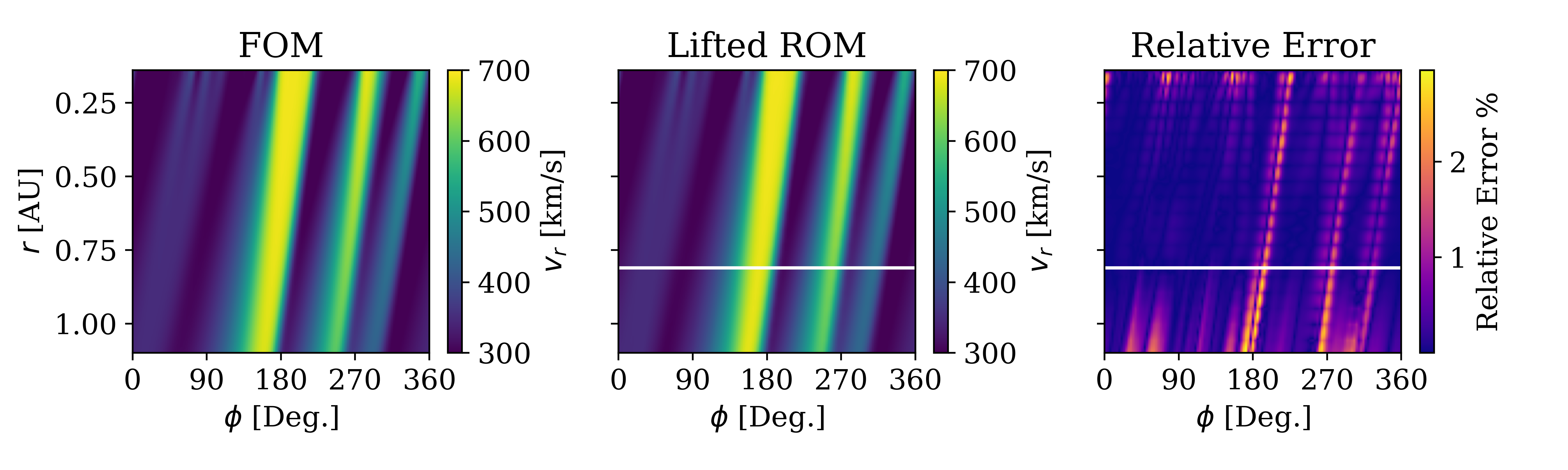}
    \includegraphics[width=1\linewidth]{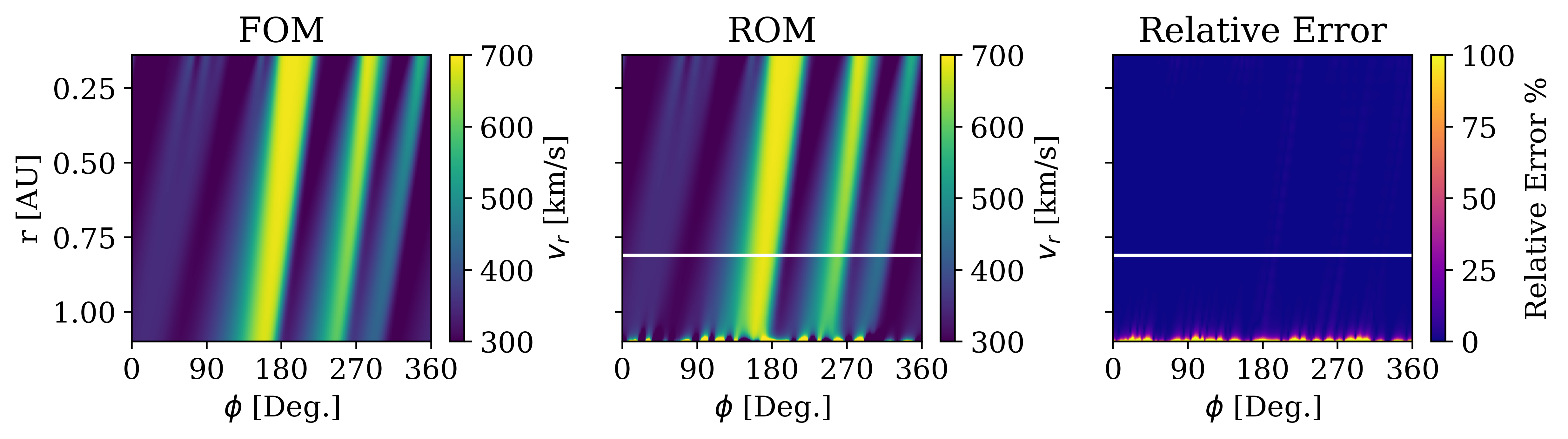}
    \caption{(Top) Solution of the quadratic ROM in the lifted variables, of the form $\dot{\bx}(r)= \bA \bx(r) + \bH\left[\bx(r) \otimes^{\prime} \bx(r)\right]$ from~\eqref{cubic-shift-lift-fom}; 
    (Bottom) Solution of the quadratic ROM in the original variables, of the form $\dot{\bv}(r) = \bA \bv(r) + \bH\left[\bv(r) \otimes^{\prime} \bv(r)\right]$. 
    }
    \label{fig:quadratic-ROM-solutions}
\end{figure}
% %%%%%%%%%%%%%%%%%%%%%%%%%%%%%%%%
%%%%%%%%%%%%%%%%%%%%%%%%%%%%%%%%
\begin{figure}
\centering
    \includegraphics[width=1\linewidth]{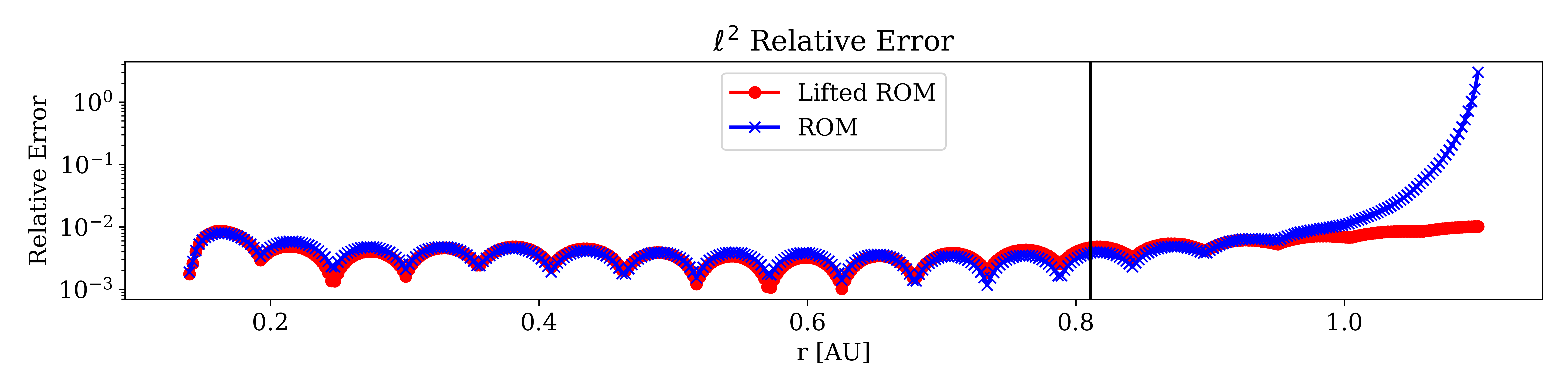}

    \caption{Relative error of the learned quadratic ROMs in both lifted variables and original variables.}
    \label{fig:quadratic-ROM-errors}
\end{figure}
%%%%%%%%%%%%%%%%%%%%%%%%%%%%%%%%

%%%%%%%%%%%%%%%%%%%%%%%%%%%%%%%%%%%%%%%%%%%%%%%%%%%%%%%%%%%%%%%%%%%
\section{Conclusions and outlook} \label{sec:conclusions}
%%%%%%%%%%%%%%%%%%%%%%%%%%%%%%%%%%%%%%%%%%%%%%%%%%%%%%%%%%%%%%%%%%%
We presented novel theory for quadratization of non-autonomous ODEs and ODEs with arbitrary dimension. In particular, we provided existence results, depending on the regularity of the input function, for cases when a polynomial control-affine system can be rewritten as a quadratic-bilinear system. In another thrust, we developed an algorithm that generalizes the process of quadratization for systems with arbitrary dimension that retain the nonlinear structure when the dimension grows. A specific example are semi-discretized PDEs, where the nonlinear terms remain identical when the discretization size increases. 
As an important aspect of this research, we extended the \Qbee software with capabilities for both non-autonomous systems of ODEs and ODEs with arbitrary dimension with rigorous optimality guarantees in the former case. 
We presented a suite of ODEs that were previously reported in the literature, and where our new algorithms outperform previously reported lifting transformations. We also highlight an important area of lifting transformations: reduced-order model learning can benefit significantly in working in the correct and optimal lifting variables, where learning quadratic models is relatively straightforward. 

We anticipate that this research will influence various disciplines and different use cases that rely on quadratization, as mentioned in the introduction. Particularly, the research will spur new directions in data-driven modeling, as novel quadratizations of nonlinear dynamical systems are discovered, which can then be used for model learning. Moreover, this research will make feasible system-theoretic model reduction for a larger class of nonlinear systems, as the optimal quadratizations are the starting point for quadratic-bilinear model reduction. 

Despite this progress, there remain several open challenges. 
First, there is a need for more theory (existence and optimality) of quadratization of non-polynomial systems of ODEs, alongside with practical and fast algorithms to find those. 
Second, an interesting question that we currently investigate is whether the symbolic computing tools used for ODEs can be carried over to the PDE setting. 
Third, since an optimal quadratization is typically not unique, a natural question is how to find optimal quadratizations with attractive numerical properties (stability, preserving equilibria, etc).

%%%%%%%%%%%%%%%%%%%%%%%%%%%%%%%%%%%%%%%%%%%%%%%%%%%%%%%%%%%%%%%%%%%
\section{Acknowledgements}
%%%%%%%%%%%%%%%%%%%%%%%%%%%%%%%%%%%%%%%%%%%%%%%%%%%%%%%%%%%%%%%%%%%
We thank the anonymous referee for the valuable suggestion to add more examples to the manuscript.

%%%%%%%%%%%%%%%%%%%%%%%%%%%%%%%
\bibliographystyle{siamplain}
\bibliography{bibliography,references-Boris}

\end{document}